\documentclass[a4paper]{article}

\usepackage[latin1]{inputenc}
\usepackage[T1]{fontenc}
\usepackage{graphicx}
\usepackage{appendix}
\usepackage{bm}
\usepackage{dsfont}
\usepackage{amssymb}
\usepackage{graphicx}
\usepackage{pifont}
\usepackage{nopageno}

\usepackage{amsmath,amsfonts,amscd,amsthm}
\usepackage{bbm}

\usepackage{hyperref}
\usepackage{algorithm}
\usepackage{algpseudocode}
\usepackage[latin1]{inputenc}
\usepackage{color}

\newtheorem{Satz}{Theorem}[section]

\newtheorem{Prop}[Satz]{Proposition}

\newcommand{\vectornorm}[1]{\left|\left|#1\right|\right|}

\usepackage{bm}
\usepackage{amsmath}
\usepackage{amssymb}  

\usepackage{dsfont}

\usepackage[normalem]{ulem}

\usepackage{dsfont}

\newcommand{\R}{\ensuremath{\mathds{R}}}		

\newcommand{\s}{\ensuremath{\mathfrak{s}}}		
\newcommand{\n}{\ensuremath{\mathfrak{n}}}		

\DeclareMathOperator*{\argmin}{argmin}		    

\newcommand{\KLD}{D_{\text{KL}}}				
\newcommand{\esi}{\hat{\Sigma}}					

\title{Explorative Data Analysis for Changes in Neural Activity}
\author{Duncan A.J.~Blythe\footnote{Department of Computer Science, Berlin Institute of Technology,Sekretariat FR 6-9
Franklinstr. 28/29
10587 Berlin}  \footnote{Bernstein Centre for Computational Neuroscience, Humboldt-Universit\"at zu Berlin
Philippstr. 13, Haus 6
10115 Berlin}, Frank C. Meinecke\footnotemark[1], Paul von B\"unau\footnotemark[1] and\\
Klaus-Robert M\"uller\footnotemark[1] \footnote{Department of Brain and Cognitive Engineering, Korea University, Anam-dong, Seongbuk-gu, Seoul 136-713, Korea
} }
\date{\today}

\begin{document}
\pagestyle{empty}
\maketitle

\begin{abstract}
Neural recordings are nonstationary time series, i.e.~their properties typically change over time. Identifying specific changes, e.g.~those induced by a learning task, can shed light on the underlying neural processes. However, such {\em changes of interest} are often masked by strong {\em unrelated changes}, which can be of physiological origin or due to measurement artifacts.
We propose a novel algorithm for disentangling such different causes of non-stationarity and in this manner enable better neurophysiological interpretation for a wider set of experimental paradigms. A key ingredient is the repeated application of Stationary Subspace Analysis (SSA) using different temporal scales. The usefulness of our explorative approach is demonstrated in simulations, theory and EEG experiments with 80 Brain-Computer-Interfacing (BCI) subjects.

 \end{abstract}

%
%
\section{Introduction}


In analysing multivariate time series, as e.g. recorded in
neurophysiological experiments, we face a challenge: artifacts,
signals resulting from different types of brain activity -- task-relevant and
task-irrelevant -- and changes thereof are observed as a
highly variable stream of data.

Some of this variability can be attributed to noise \cite{EEG_chaos}, some to
learning and plasticity \cite{Mandelblatt:2011q} or also to unknown latent variables \cite{neuralVariabilityPremotor}; in other
words the data exhibits changes on many scales. Explorative data
analysis methods such as PCA \cite{originalPCA}, ICA \cite{Comon1994287} and other projection methods can
contribute to disentangling confounding trends from the data \cite{LinearEEG}. However they generally assume an
underlying stationary distribution of the data, so if distribution
changes occur, say, due to learning processes or due to lapses of
attention, they will yield suboptimal results.


This paper addresses the question of finding \emph{changes} in
multivariate neural data which occur over time; we propose a model, that enables for the
first time the extraction of both global behavioral condition changes as
well as experimental condition specific changes. 
To further illustrate this abstract scenario, a concrete example from
the field of Brain Computer Interfacing is as follows; the BCI
user is instructed to perform a certain number of motor imagery
commands, for example, imaginations of left hand, right hand and foot
movements. Because the user \emph{learns} to use the BCI over time, the
data in each of these conditions may change over time: these are
changes \emph{specific} to each of the conditions\footnote{Users can
  in fact learn different types of imaginations at different speeds,
  similarly to the ability of a tennis player to learn his fore- and
  backhand strikes to differing levels of proficiency}. Moreover, often users become
tired while using the BCI. This results in higher levels of alpha
activity ($\approx$7-15Hz) present in \emph{all} conditions as the
experiment progresses \cite{oai:eprints.pascal-network.org:3317}. This increase in alpha would correspond to a
\emph{background} change in neural activity.  We would like to model
these two types of changes separately so they may be examined in
isolation.

%

In principle, one might think that the task of finding changes in the
time domain of a dataset has already been solved; one may simply optimize an
objective function which measures change over time on some subset of
the data, as was proposed for Stationary Subspace Analysis (SSA)
\cite{PRL:SSA:2009:DISCRIM}. However, in practice, these changes may not
\emph{all} be relevant to the data analytic task.  This is illustrated
in Figure~\ref{fig:illustration}; suppose the data is split into
experimental conditions and we are interested in changes which are
particular to only \emph{one} of these conditions. The first time
series displays data from all conditions with the data from the target
condition of interest highlighted. We see that the condition of
interest displays a weak change, not contained in the other
conditions. However, we also observe that there are other
\emph{stronger} changes (for example, related to artifacts) in all
conditions. Fortunately, the weaker change of interest, which is
specific to the highlighted condition, has a different \emph{origin}
than the changes affecting all conditions: as we will see below, this
situation will play an important role in our approach below\footnote{In high dimensional data setups, that these different types of changes have the same origin has very low probability. Moreover, neurophysiological considerations imply that the two sources of change are more likely to originate from different points in the brain.}.

Scenarios where similarly complex changes are hidden in the data arise 
in many experimental paradigms in neurophysiology \cite{timeRelatedAudio, BCIadapt}. So far this fact has
been ignored since there was no sound algorithmic solution to address
the issue. The presented framework therefore aims to contribute to a
better understanding of complex experimental paradigms which involve
distributional changes \emph{by construction}, e.g. paradigms focusing on
learning and plasticity \cite{Mandelblatt:2011q}.

We present our new model more formally in Section~\ref{sec:model} and
outline a method for removing the background changes in
Section~\ref{sec:public_method}, for which we provide theoretical
guarantees, in Section~\ref{sec:con_proof}. We then evaluate the
quality of the estimation method in simulations, in
Section~\ref{sec:simulations}.  Lastly, the analysis of the BCI data
sets is covered in Section~\ref{sec:application}.

\begin{figure}
\begin{center}
\includegraphics[width=120mm,clip=true, trim=1mm 52mm 4mm 10mm]{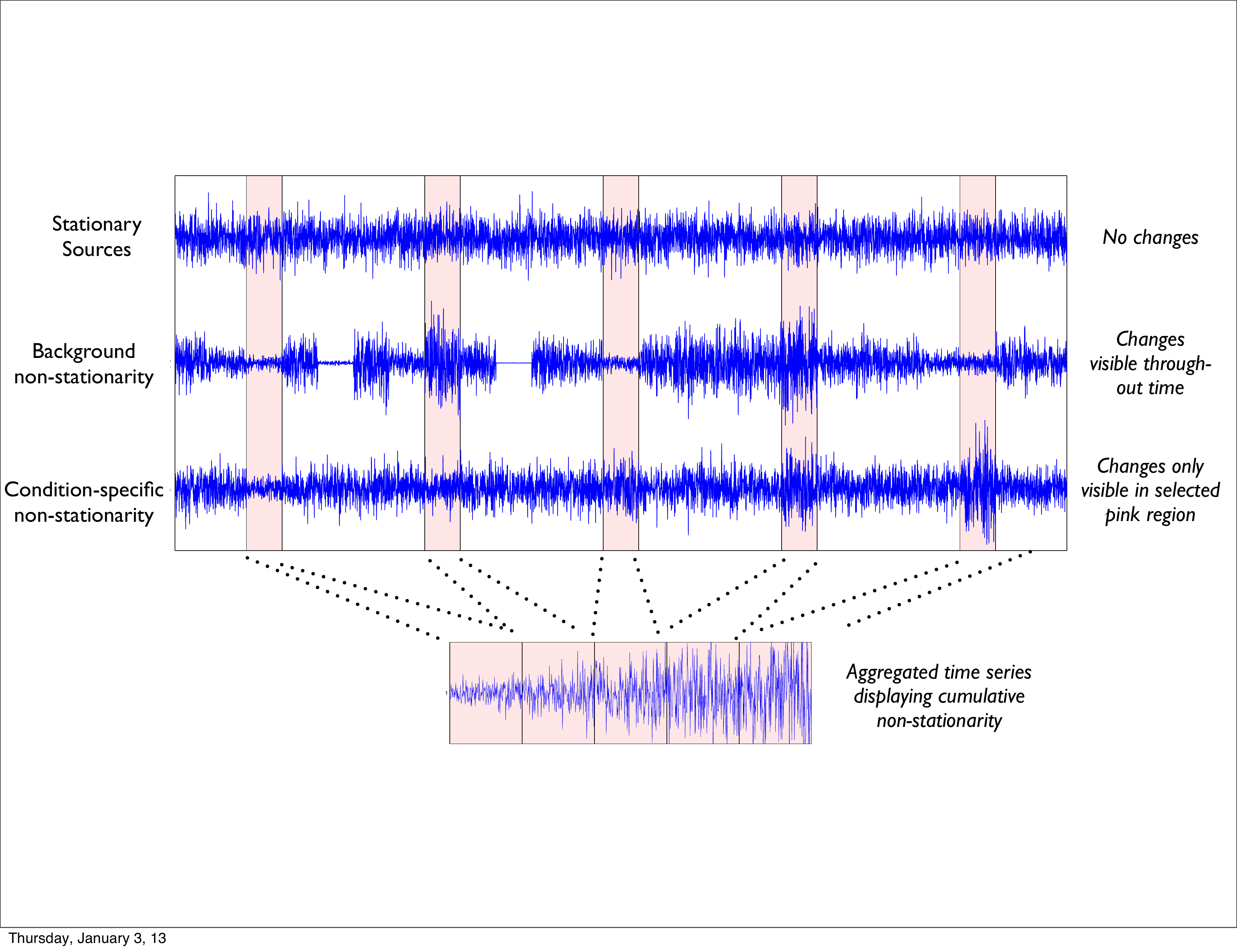}
\end{center}
\caption{The figure illustrates the problem setting: changes specific to the data in one condition are masked by stronger changes taking place over all conditions. The top time-series displays the data from all conditions, with a selected target condition marked in red. 
\label{fig:illustration}}
\end{figure}

\section{Model and Parameter Estimation}
\label{sec:model}

In this section, we introduce the generative model underlying our approach and a method for estimating its parameters, 
as well as theoretical results that provide us with performance guarantees.

\subsection{Generative Model}
\label{sec:mod}

The generative model is an extension of the SSA linear mixing model for a multivariate non-stationary time series  \cite{PRL:SSA:2009:DISCRIM}\footnote{Related mixing models often applied in neuroscience include the Independent Component Analysis model (ICA) and the Principal Component Analysis model (PCA). For background on linear mixing models and source separation see \cite{Comon1994287} and \cite{PRL:SSA:2009:DISCRIM}.}. 
In this model, the observed multivariate data, $\mathbf{x}(t)$, is generated as a linear superposition of two groups of latent 
(only indirectly observable) sources. The first group, $\mathbf{s^\s}(t) = [s_1(t), \dots, s_d(t)]^\top$, are stationary, and the 
second group, $\mathbf{s^\n}(t) = [s_{d+1}(t), \dots, s_{D}(t)]^\top$, are non-stationary. In this context, stationarity of a 
time series is defined as the first two moments (mean and spatial covariance) being constant over time; so-called "weak"-stationarity is usually defined to include the temporal cross-covariance being constant over time, however, in the current contribution we concentrate on the spatial covariance (at time-lag $0$)\footnote{In this publication we consider time-series with vanishing autocorrelations; thus wherever "covariance" is printed, we mean to refer to a spatial covariance.}. 
Moreover, for the purpose of our final EEG analysis (see Section~\ref{sec:application}), we will ignore the mean (since the data are high pas filtered) and we will focus only on changes in the covariance matrix, which equates
to the signal power when projected to a univariate time-series. 
The observed $D$-variate data, $\mathbf{x}(t)$, is generated as a linear transformation of the two groups of stationary and non-stationary sources by an
unknown square mixing matrix $A$, given in terms of the $d \times D$ and $(D-d) \times D$ rectangular matrices $A^\s$ and $A^\n$.  
\begin{eqnarray}
	\mathbf{x}(t) &=& A\mathbf{s}(t) \nonumber \\
					&=&  \begin{bmatrix} A^\s& A^\n \end{bmatrix} \begin{bmatrix} \mathbf{s^\s}(t) \\ 
						  \mathbf{s^\n}(t) \end{bmatrix} \label{eq:SSA}.
\end{eqnarray}
The matrix $A$ is called mixing matrix because the entry $A_{ij}$ determines how the $j$-th source contributes to the 
$i$-th dimension of the observed data. If the sensors corresponding to the measurements are spatially distributed, 
then the columns of the mixing matrix can be interpreted visually as patterns, e.g.\ scalp maps in EEG analysis.
The columns of the matrices $A^\s$ and $A^\n$ span the stationary and non-stationary subspaces respectively. 
Note that the column spaces of these matrices need not be orthogonal.
The aim of SSA is to invert this mixing, i.e.\ to find a demixing matrix which allows us to recover the two 
groups of sources from the recorded mixed signals. 

\begin{figure}
	\centering
	\includegraphics[width=120mm,clip=true,trim=45mm 100mm 45mm 35mm]{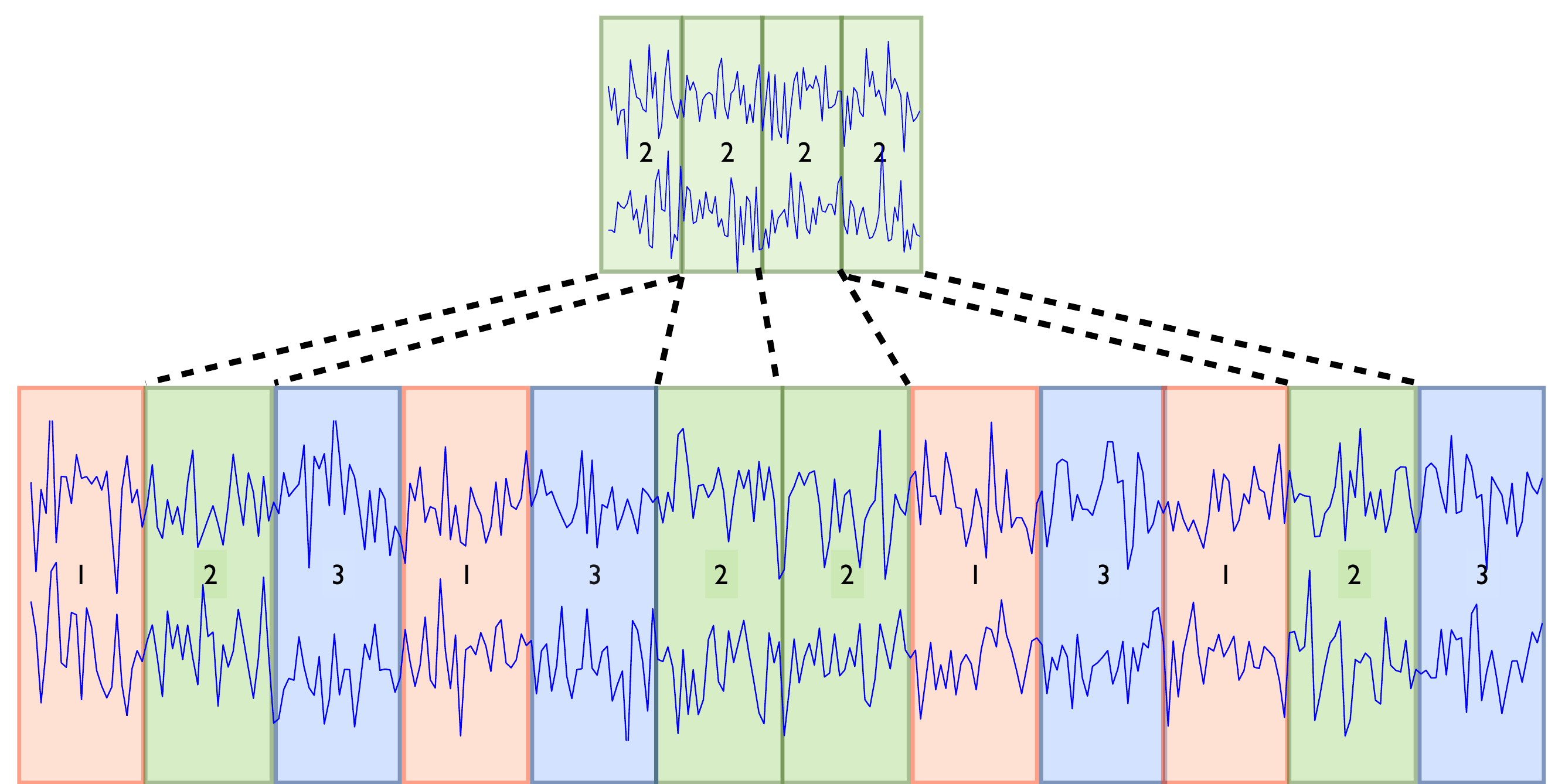}
	\caption{The figure displays the block trial structure of the multivariate data which we model. Trials recorded in separate behavioural conditions are depicted as colored blocks. A trial from any 
condition occurs with equal probability. Above is illustrated how the data from the 2nd condition is formed.
	\label{fig:trial_illustration}}
\end{figure}

However, our setting differs from the SSA model (given in Equation \ref{eq:SSA}) in that we assume that the observations 
are generated as the \textit{sum of two latent data generating systems}: a background system, and a condition-specific 
system---where the latter contains the information of interest and the former introduces irrelevant distribution
changes which we want to remove. In line with these assumptions, we model several time-series
corresponding to each experimental condition; this is at first glance inconsistent with all
data being drawn in practice from a single empirical time-series (corresponding to the samples recorded from the measurement device) but in fact,
this correctly captures the fact that each condition is subject to separate parameters.

In order to obtain these empirical time-series from the recorded data we assume that for each condition, 
several trials (for example corresponding to a single stimulus presentation or motor imagery etc.) are recorded; moreover, we assume that conditions occur with equal probability in the sequence of 
trials\footnote{Randomization of conditions is standard practice in many experimental designs.}. 
The trials may assumed to be approximately stationary (although this is not essential) and non-stationarity
exhibits itself between trials. Then the data from the $k-$th experimental condition may be obtained by concatenating
all of the trials recorded in that condition.

Figure~\ref{fig:trial_illustration} illustrates the setup of the data. For instance, in a typical EEG motor 
imagery brain computer interface (BCI) session two conditions are recorded: left imagery and right imagery. 
A trial for one of these conditions consists of the data recorded during a single imagined movement, for example 
an imagined left hand movement in the left condition\footnote{Note that in practice the trials in each condition may have different lengths. This will generate a bias in estimation if not appropriately  dealt with. One possible solution is to subsample segments of the longer condition conditions so that all trials are forced to be of the same duration. This implies that there will be sections of unused data where perhaps beforehand there were not. This is not an obstacle to the application of our method.}. 

These considerations allow us to plausibly model the observed data $\mathbf{x}_{k}(t)$ in condition $k$ at time $t$ (within condition time-index) as 
generated by the sum, 
\begin{eqnarray}
\label{eq:model}
	\mathbf{x}_{k}(t) &=& A\mathbf{s}(t) + B_k\mathbf{r}_k(t) \\
	 &=& \begin{bmatrix} A^\s & A^\n \end{bmatrix}  \begin{bmatrix} \mathbf{s^\s}(t) \\ \mathbf{s^\n}(t) \end{bmatrix} + \begin{bmatrix} B_k^\s & B_k^\n \end{bmatrix}  \begin{bmatrix} \mathbf{r_k^\s}(t) \\ \mathbf{r_k^\n}(t) \end{bmatrix}
\end{eqnarray}
where $A\mathbf{s}(t)$ is the contribution of the background system and $B_k\mathbf{r}_k(t)$ is the condition-specific
contribution \footnote{Similarly to the standard SSA model,
$A^\s$ and $A^\n$ are of size $d \times D$ and $(D-d) \times D$ and $B_k^\s$ and $B_k^\n$ are of size $d_k \times D$ and $(D-d_k) \times D$ corresponding to the numbers of stationary sources in each system.}.
The corresponding inverses are $A^{-1} = P = \begin{bmatrix} P^\s \\ P^\n \end{bmatrix}$ and $B_k^{-1} = Q_k = \begin{bmatrix} Q_k^\s \\ Q_k^\n \end{bmatrix}$. These matrices stay constant over time, whereas the sources, $\mathbf{s}(t)$ etc. change over time. Moreover, the matrix $A$ is constant over all $K$ conditions, whereas, $B_k$ is a variable which depends on the condition.
The superscripts $\s$ and $n$ refer to the stationary and non-stationary contributions of each system; as stated above, stationarity refers to the \emph{mean} and \emph{covariance} remaining constant over time (not the time-courses).

Note in particular at this point that $\mathbf{s}(t)$ (and $\mathbf{r}_k(t)$) are to be viewed as random time-series variables; this implies that 
while $\mathbf{s}(t)$ represents the \emph{same} random variable in two given conditions, 
the sample time-series drawn from these random variables in two sample experimental conditions will be \emph{different} entities. However, the time-index $t$ should be seen in each time series $\mathbf{x}_{k}(t)$
to represent the \emph{same} quantity in order that we may, in practice, consider approximate simultaneity (relative to the time-scale of the experiment) of samples drawn from separate experimental 
conditions.


The aim of our analysis is to remove the background non-stationary components $\mathbf{s^\n}(t)$ from 
the signals in order to be able to analyze the distribution changes in a particular condition-specific
system. To this end, we want to find a linear projection to the $d$ stationary components of the background 
system. Even though in the final analysis we will only look at the data from one of the conditions, we use the 
whole data set for estimating this projection as we will see in the next section.

\subsection{Parameter Estimation}
\label{sec:public_method}

The parameters of the model (given in Equation \ref{eq:model}) are the mixing matrix $A$ of the background system, and 
the mixing matrices $B_1,\ldots,B_K$ for each of the $K$ condition-specific systems. For the purpose of 
removing the background non-stationary components, we are only interested in finding a matrix 
$P^\s \in \mathbb{R}^{d \times D}$ which, when applied to the observed data, allows us 
to remove the non-stationary components of the background system, i.e.,
\begin{align*}
	P^\s \mathbf{x}_{k}(t) & = P^\s A^\s \mathbf{s^\s}(t) + \underbrace{ P^\s A^\n }_{=0} \mathbf{s^\n}(t) + P^\s B_k\mathbf{r}_k(t) \\
							& =	 P^\s A^\s \mathbf{s^\s}(t) +  P^\s B_k\mathbf{r}_k(t), 
\end{align*}
because the rows of $P^\s$ are in the dual space of the non-stationary subspace spanned by $A^\n$ such that 
$P^\s A^\n = 0$. Thus, the only matrix referred to by the model (\ref{eq:model}) that we estimate is 
a basis for the non-stationary subspace. 

The projection matrix $P^\s$, which recovers the stationary sources $\mathbf{s^\s}(t)$, is estimated by applying SSA \cite{PRL:SSA:2009:DISCRIM} in a particular way\footnote{Note that we will refer to $P^\s$ as the "projection to the background stationary sources"; this overloads a closely related use of "projection matrix", referring to a square matrix $A$ with $A^2 = A$. The
matrix obtained by filling out $P^\s$ to a square matrix with zeros in the final rows is a projection matrix in this second sense.}. In brief, SSA finds the 
projection to the $d$ stationary sources by minimizing the difference in distribution across epochs of 
the observed data. Special care is required here, however, since although we modeled the $K$ conditions separately in Equation~\ref{eq:model}, for the purpose of estimation of $P^\s$ we consider \emph{all} conditions as
resulting from the same empirical time-series.
That is, the raw time series obtained by concatenating all trials from all conditions, in the order they appear (see Figure~\ref{fig:epoch_illustration}), is divided into $N$ epochs (using a chronological segmentation or a 
sliding window) for which epoch covariance matrices $\esi_1, \ldots, \esi_N$ are estimated\footnote{The choice of epochs is a delicate issue. The same issues for our problem setting apply as for the SSA problem setting and details may be found in the SSA expository paper \cite{PRL:SSA:2009}. Whether the epoch length affects the stationarity of the data is an application-bound issue. In our BCI application below, we choose epochs so as to span several trials, thus reducing noise levels and including multiple instances of each condition but short enough so as to conform to the bounds derived in the mentioned paper \cite{PRL:SSA:2009}.}.
The optimal 
stationary projection $\hat{P}^\s$ is found by minimizing the difference between the projected epoch covariance
matrices, as measured by the Kullback-Leibler (KL) divergence $\KLD$ between Gaussians. The SSA objective function 
is the sum of the KL-divergences between each epoch and the average epoch, whose mean and covariance we 
can set to $0$ and $I$ respectively without loss of generality. This leads to the objective function,  
\begin{align*}
	\mathcal{L}\left(P \right) = \sum_{i=1}^N \KLD \left[ \mathcal{N}(0, P \esi_i P^\top) \, || \, \mathcal{N}(0,I) \right], 
\end{align*}
which is minimized using an iterative gradient-based procedure to find $P^\s$, 
\begin{align*}
	\hat{P}^\s = \argmin_{P P^\top = I} \sum_{i=1}^N -\log \det \left(  P \esi_i P^\top  \right) . 
\end{align*}

 \begin{figure}
	\centering
	\includegraphics[width=120mm]{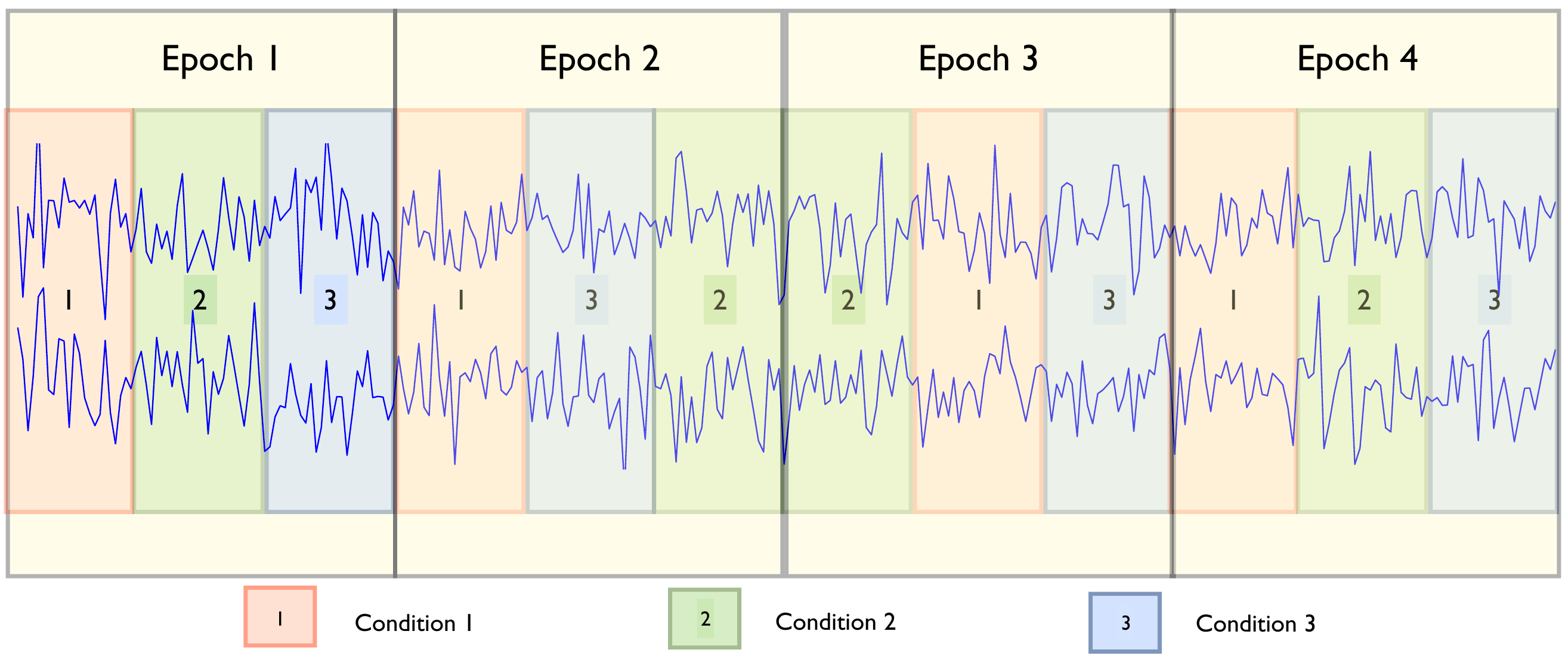}
		\caption{SSA epoch structure of the multivariate empirical time-series for estimating of the background stationary projection. 
	     Each epoch contains equal numbers of trials from each condition. 
	\label{fig:epoch_illustration}}
\end{figure}
 
Although SSA was originally developed to invert the model (\ref{eq:SSA}), it is also applicable to our problem setting. This is because the weaker the changes in the condition-specific non-stationary sources 
$\mathbf{r}_k^\n(t)$, the more closely the model (\ref{eq:model}) is approximated by the SSA model. 
To achieve this, we design our epochs in such a way that each epoch contains the same number of 
trials from each of the $K$ conditions; see Figure~\ref{fig:epoch_illustration} for an illustration.
Thus we obtain SSA epoch covariance matrices which are approximately equal to the covariance matrix of the background system up to an approximately \textit{stationary} (i.e.\ constant over SSA epochs) contribution 
from the condition-specific systems. More formally, the true covariance matrix (covariance at lag $0$ with size $D \times D$) of the 
$i$-th SSA epoch $\Sigma_i$ is given by, 
\begin{align*}
	\Sigma_i = C_i + \bar{R}_i
\end{align*}
where $C_i$ is the covariance matrix of the background system, and $\bar{R}_i$ is the average of
the covariance matrices over all conditions in the $i$-th SSA epoch, 
\begin{align}
	\bar{R}_i = \frac{1}{K} \sum_{k=1}^K R_{i,k},  \label{eq:av_epoch_cov_classes}
\end{align}
with $R_{i,k}$ being the covariance matrix of the $k$-th condition when considering only the condition specific system in the $i$-th SSA epoch. 
As we will show in the next section, the condition-specific part $\bar{R}_i$ converges 
to stationarity in the sense that it becomes constant over SSA epochs. This means that SSA
will correctly identify the non-stationary directions in the background system which we 
want to remove.

\subsection{Theoretical Results}
\label{sec:con_proof}

In this section we describe a theoretical guarantee for the quality of the approach outlined in the previous 
section. The rather technical proof is deferred to Appendix~\ref{app:proofs}; here we focus on the intuitive 
interpretation of the result. 
In particular, Proposition~\ref{prop:converge}, which we prove in the Appendix, states that the covariances resulting from forming SSA epochs containing several conditions will converge to stationarity in the limit of the numbers of conditions $K$ and condition-specific stationary dimensions $d_k^\s$.
This allows the SSA algorithm to identify the desired part of the background system, because the only remaining factor affecting estimation are then the small sample effects affecting SSA.

\begin{Prop}
Let $\bar{R_{i}}$ be, as defined in Equation~\ref{eq:av_epoch_cov_classes},
the average condition specific covariance in the $i$-th epoch, $K$ the number of conditions, $N$ the number of epochs, $D$ the number of dimensions, where the number of non-stationary dimensions, $d_n$, is fixed. 

 Assume, that the non-stationarity is bounded, i.e.~for all $i,k$, corresponding to the $i$-th epoch and $k$-th condition, $\vectornorm{R_{i,k} - \frac{1}{N}\sum_{i=1}^NR_{i,k}}_\mathcal{F} < \Gamma$, for some constant $\Gamma$, where $||\cdot||_\mathcal{F}$ is the Frobenius norm on matrices. Assume finally that the columns of the matrices $B_k^\n$, the condition-specific non-stationary subspaces,  are also uniformly distributed with orthonormal columns\footnote{They are a sample of Haar measure.}. 
 
 Then let $P$ be any projection matrix. Then, for any $\delta > 0$ we have the following relation in probability, in $K$ and $D$, of $\bar{R}_i$ to the average covariance over epochs $\frac{1}{N}\sum_{j=1}^N \bar{R}_j$ under the projection $P$, 
 
\begin{multline}
	\text{Pr}\left(\frac{1}{N}\sum_{i=1}^N\vectornorm{P^\top\left( \bar{R_{i}} - \frac{1}{N}\sum_{j=1}^N  \bar{R}_{j}\right)P}_\mathcal{F} \right. \\
\left.	> \mathcal{O}\left(\frac{1}{K}\right) + \mathcal{O}\left(\frac{1}{\sqrt{D}}\right) + \mathcal{O}\left(\frac{1}{\sqrt{\delta \sqrt{D}}}\right)\right) 
		<\delta.
\end{multline}
\label{prop:converge}
\end{Prop}
\begin{proof}
See Appendix \ref{app:proofs}.
\end{proof}

The term on the left of the inequality expresses the non-stationarity
measured in the terms $\bar{R}_i$ over epochs; each term of the sum
expresses the deviation of $\bar{R}_i$ from the average over SSA
epochs.  Note that this is a different measure from the measure of
non-stationarity that we use for \emph{optimization} with SSA: the SSA
objective function is invariant to changes in the parametrization of
the data space, which is desirable for estimation, whereas the measure
in the proposition is not.  However, for the purpose of deriving a
convergence rate, the measure in the proposition is more convient. Nevertheless both measures are equivalent in the sense that each
is zero if and only if the other is zero.

The term on the right of the inequality decreases in $K$ and $D$: for
complete convergence to be guaranteed under the assumptions of the
proposition, note that both $K$ (number of epochs) and $D$ (number of
dimensions) should be assumed to grow. Nevertheless, either $D$ or $K$
increasing guarantee a better estimate. That $D$ may assumed to be
much bigger than $d_n$ in practice is plausible, as often
non-stationarities are confined to low dimensional subspaces, for
example, artifacts in EEG are often confined only to electrodes at the
edge of the scalp.

Note also that the proposition is conditioned on two important assumptions: firstly that the non-stationary subspaces of the condition specific systems are sampled uniformly at random and secondly that 
the non-stationarity is uniformly bounded over classes ($\vectornorm{R_{i,k} - \frac{1}{N}\sum_{i=1}^NR_{i,k}}_\mathcal{F} < \Gamma$).
These are plausible assumptions to make on the model. The first assumption is equivalent to assuming we have no prior knowledge of the condition specific systems's mixing matrices, but may be weakened without making convergence impossible and the second assumption is necessary for estimation to occur at all, but is nevertheless plausible; for example, in the BCI application (below), the condition specific non-stationarities are much weaker than 
the background non-stationarities.




\section{Simulations on Synthetic Data}
\label{sec:simulations}

In this section, we evaluate the proposed method in controlled simulations on synthetic 
data, where we can objectively assess the performance based on an artificial ground truth.
Our analysis consists of two parts. First of all, we investigate how accurately our algorithm 
can identify the true background stationary components. In our model, these are the 
relevant components for subsequent analysis. In the second part of this section, we show 
that these components are useful. To that end, we apply our method as a pre-processing step to 
change-point detection in a situation where there are irrelevant background non-stationary 
components, which we want to remove. We analyse the relative merits of our pre-processing 
in different scenarios.

\subsection{Data Generation}
\label{sec:data}

The synthetic data that we use in the controlled simulations is generated according to our 
model (Equation~\ref{eq:model}). At each time point, the observed $D$-variate sample $x_k(t)$ 
in the $k$-th condition is a sum of contributions from the background and the 
condition-specific system. The background variables are linearly transformed by a 
random orthogonal mixing matrix $A \in \R^{D \times D}$ that is kept fixed over all conditions. 
The condition-specific components are linearly transformed by an orthogonal mixing matrix 
$B_k \in \R^{D \times D}$ that is chosen randomly for each condition $1 \leq k \leq K$, where 
$K$ is the total number of conditions. In the simulations, the number of stationary and 
non-stationary components is kept fixed to the same value for both systems.

The background and the condition-specific system each comprises
two groups of components: stationary and non-stationary. The mean and
covariance matrix (at time-lag $0$) of the stationary components are the zero vector $0$
and the identity $I$ respectively. The mean of the non-stationary
components is also fixed to zero, whereas the covariance matrix (again at time-lag zero) 
in each epoch is chosen according to two Markov models, one for the background
and one for the condition-specific system. Both Markov models consist
of five states corresponding to five covariance matrices. The
covariance matrices are diagonal with entries drawn randomly
from a set of log-spaced values over the interval $[1/\nu, \nu]$ with $\nu >
1$; thus a high $\nu$ corresponds to a high level of non-stationarity. For the background system, this parameter is
$\nu_\text{shared}$ and for the condition-specific it is
$\nu_\text{specific}$.  The probability of changing the state in the Markov model
(i.e.\ switching to another covariance matrix) is parametrized by
$p_\text{shared}$ and $p_\text{specific}$ (denoting precisely the probability of remaining in the current state; the probability if changing to each other state is the same) for the background and
the condition-specific system respectively. Thus a low $p$ corresponds to few changes occurring. The shortest segment of stationary data possible consists of $50$
samples, generated from a Gaussian with the respective covariance
matrix; thus changes may occur at most once every 50 data points.

To sum up, the key parameters of interest are the strength of the non-stationarity in both
systems ($\nu_\text{shared}$ and $\nu_\text{specific}$), and the probability of changes in the
covariance matrix of the non-stationary sources ($p_\text{shared}$ and $p_\text{specific}$). 

\subsection{Finding the background Stationary Sources}
\label{sec:get_pub}

In the first set of experiments, we investigate how well our method can identify the true stationary
components of the background system. We analyze the influence of the number of different conditions 
in two scenarios: low and high condition-specific degrees of non-stationarity. 
The total number of dimensions is set to 20 and the number of non-stationary directions is set to 
1 and 2 for the condition-specific and background system respectively; the total number of sampless
is 10000. The error is measured as the subspace angle between the true non-stationary subspace
of the background system and the estimated one. 

\begin{figure}
	\centering
	\includegraphics[width=120mm]{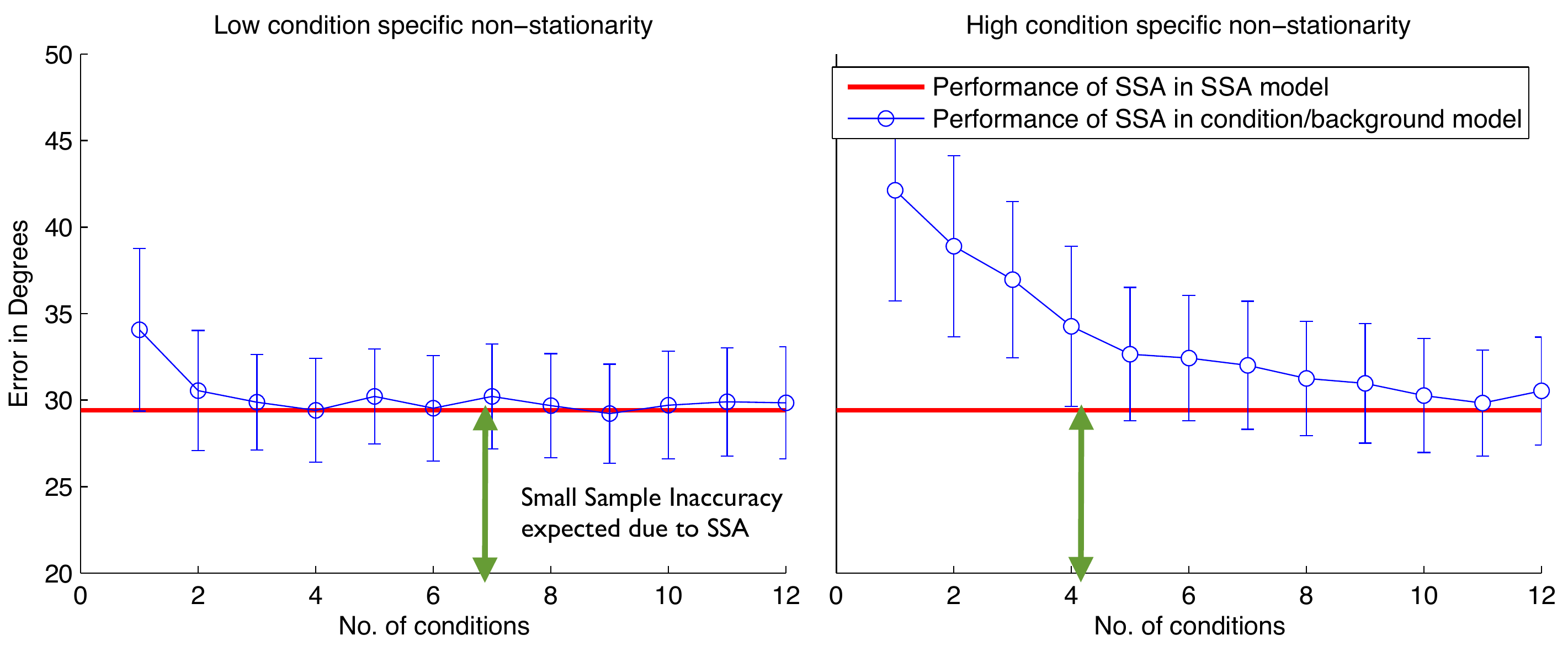}
	\caption{
	Performance of the proposed method measured in terms the subspace angle (vertical axis) between the true
	and the estimated non-stationary subspace of the background system, for different numbers 
	of conditions (horizontal axis). The left and right panels shows the results for the cases where the
	condition-specific non-stationarity is low and high respectively.  In each case, the red line denotes the case in which the condition-specific systems are replaced with stationary noise $d_k^\s = D$ as a baseline result. The red line case thus represents the limit in potential accuracy of the method relative to accuracy in the standard SSA model.
\label{fig:acc_sim}
}
\end{figure}

The results are shown in Figure~\ref{fig:acc_sim}. In the left panel, we see the results for low 
condition-specific non-stationarity. As we increase the number of conditions, the error in finding the 
true stationary components (blue line) becomes smaller. This is because the non-stationary contributions of the 
condition-specific systems become more equally spread across the whole $D$-dimensional space since for 
every condition $k$ we choose a random mixing matrix $B_k$. This allows us to distinguish more clearly 
between the background and the condition-specific non-stationary directions and, in turn, find 
the background stationary components. The red line shows the error on the same data where we 
have removed the condition-specific contributions: this is the lowest possible error level on the condition specific model for finite samples. We see that as soon as we have at least four different conditions, 
we can identify the true background stationary components as well as if there were no 
condition-specific non-stationarities.

The right panel shows the results in the case where the degree of non-stationarity in the condition-specific 
system is high relative to the background system. As we can see, this makes the desired stationary components 
of the background system more difficult to identify than in the previous case, because the 
overall distribution is dominated by changes in the condition-specific system. We therefore need a larger number 
of conditions (around 11) in order to distinguish between the non-stationarity of the two systems.

\subsection{Identification of the background Non-Stationary Components in Simulated EEG}

We illustrate the application of SSA to artificial data generated by a
realistic EEG forward model.
The head model we use was calculated on the basis of the MRI scans of
152 human participants \cite{Kopfmodel1}, the electrodes were placed
on the head according to the international 10-5 system and the
forwards mapping from simulated dipoles to voltages at the electrodes
were computed using the semi-analytic methods of the contribution of
Nolte et al. \cite{LeadFields}.
Given this framework we simulated two background non-stationary
sources, the first 2 cm underneath the electrode T7 close to the edge of the montage, to simulate the effects of a non-stationary 
EMG component and the second at the primary 
visual cortex, 2cm under the electrode Oz, to simulate the effect of a non-stationary alpha source.

Subsequently we modeled three scenarios.
In the first scenario, for each of 13 conditions, in addition to the two background non-stationary sources, a non-stationary source was positioned
at random in the brain together with 12 stationary sources, thus the dataset has dimensionality 15, within a 118 (the number of electrodes) dimensional space. The non-stationarities were generated using a Markov model, giving rise to segments of simulated EEG data 
with varying power.
The data were preprocessed using Principal Components Analysis (PCA) \cite{Pearson1901Lines} to eliminate the degenerate directions (PCA then yields a projection $U \in \mathbb{R}^{15 \times118}$ and a mixing matrix $V\in \mathbb{R}^{118 \times15}$). We then computed the estimated background non-stationary components using SSA 
and the estimated condition specific components, for the first 2 conditions. More precisely, $A^\n$ and $P^\s$, were computed on the PCA-preprocessed data using SSA; thus the scalp pattern corresponding to the most non-stationary background sources is given by composing $A^\n$ (of size $15 \times d^\s$) with the mixing matrix yielded by PCA $V$ (of size $118 \times 15$) i.e. we compute ($V \times A^\n$). Subsequently SSA was applied to the condition-specific data preprocessed using PCA and projected by $P^\s$ (i.e. on $P^\s U \mathbf{x_k}(t)$) in order to obtain a mixing matrix (which we will call $C^\n_k$) corresponding to the non-stationary condition-specific patterns; thus $B^\n_k = A^\s \times C^\n_k$.  The final corresponding scalp pattern is given by $V \times A^\s \times C^\n_k$.
The results of this first scenario are displayed in Figure~\ref{fig:ForwardModel} and show that SSA successfully locates the EMG and alpha component, and succeeds in obtaining asymmetric non-stationary components between the first two
classes, which match the true non-stationary condition specific patterns. Moreover, we see that using SSA separately on each class does not deliver the background non-stationarities and the condition specific non-stationarities successfully, whereas estimation of the background non-stationarities over conditions using SSA is successful as described in Section~\ref{sec:public_method}, following which the condition specific non-stationary patterns are successfully
found. 

Although, in general the background non-stationarities may only be reliably obtained given many conditions, we may still consider whether using SSA to remove background non-stationarities works
when we have only two conditions and under which circumstances. This reflects the situation in our application to EEG data, where we consider motor-imagery of two separate movements, described in Section~\ref{sec:application}. Thus, in the second scenario, we reduce the number of conditions to two and simulate background non-stationary sources which are \emph{stronger}
in their non-stationarity than the condition specific non-stationary sources. All other settings are kept constant. The results of this second scenario are displayed in Figure~\ref{fig:ForwardModelTwoClasses}
and show that in this case, when the background non-stationarities are strong, then these, and subsequently the condition specific non-stationarities, may be obtained nonetheless, using the proposed method.

However, in the third and final scenario we see that when the number of conditions is two and the background non-stationarities are weak, then obtaining patterns using SSA on each condition separately which
nevertheless agree across the conditions, and thus may be assumed to correspond to background non-stationarities, occurs only after the most non-stationary components have been removed. Note that we remove more non-stationary directions here than actually are generated; this is achieved simply by setting the $d$ parameter of the SSA computations to a different value than the true value. The results of this third
scenario are displayed in Figure~\ref{fig:ForwardModelSpike}.
We will
see, in the analysis of the BCI data in Section~\ref{sec:application}, that both of these second and third scenarios may arise.

\begin{figure}
\begin{centering}
\includegraphics[width=114mm,clip=true,trim=1mm 80mm 105mm 70mm]{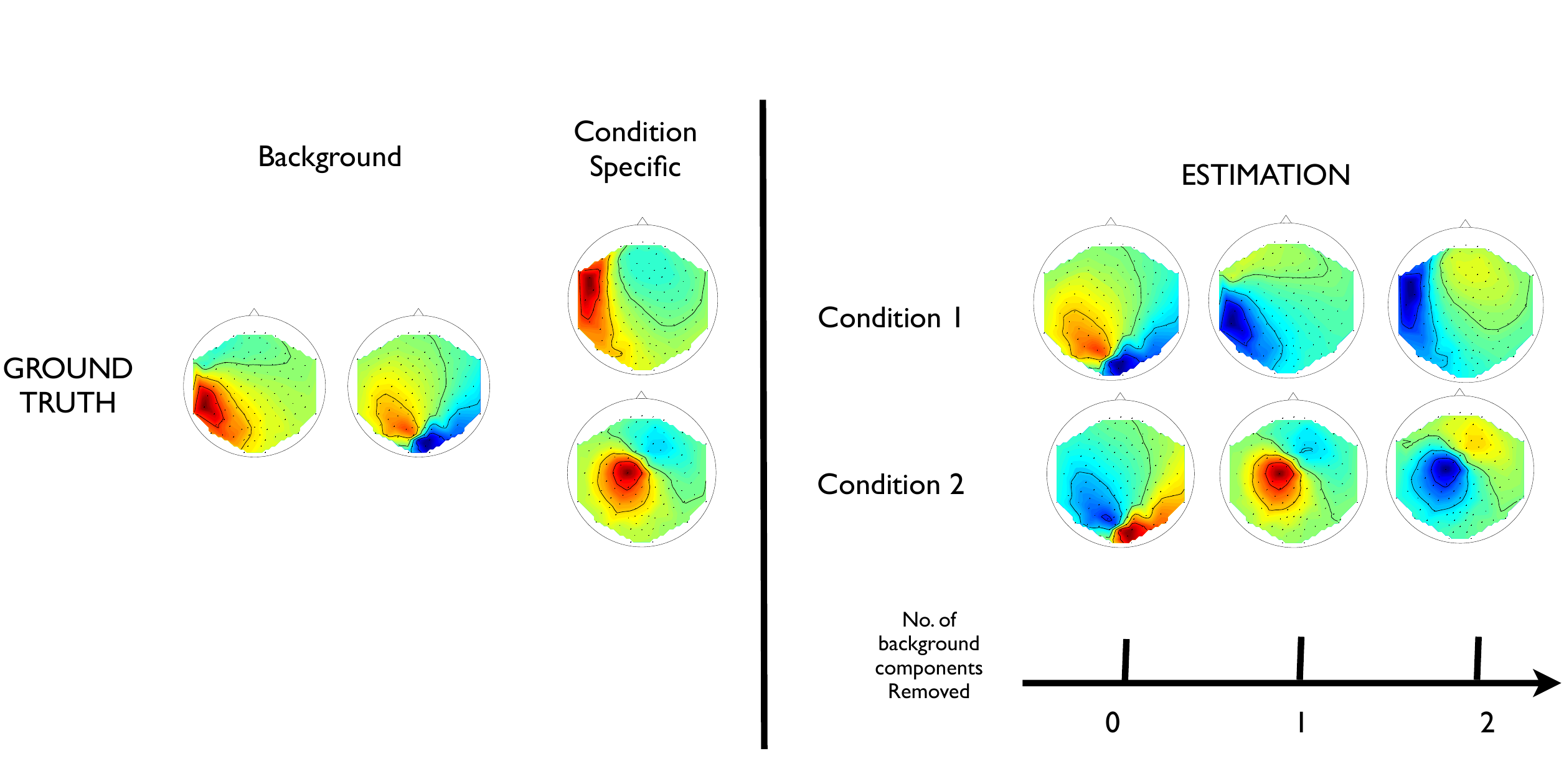}
\caption{The figure displays scalp plots computed using SSA on the EEG forward model data. The left panel displays the true non-stationary patterns, and the right hand panel displays
the estimated condition specific non-stationary patterns after removal of zero, one and two background non-stationary directions. The plots show that SSA successfully obtains the patterns resulting
from the background non-stationary alpha and background non-stationary EMG, and then allows one to concentrate on the condition specific non-stationarities in each class. Note that
after removing one background non-stationary direction, SSA finds a difference between the conditions. This is because, for the two conditions studied, the condition specific non-stationarities
are strong. Nevertheless, SSA successfully obtains the correct background non-stationarity, so that the final condition specific non-stationarities are correct.
 \label{fig:ForwardModel}}
\end{centering}
\end{figure}

\begin{figure}[h]
\begin{center}$
\begin{array}{cc}
\includegraphics[width=152mm,clip=true,trim=1mm 87mm 10mm 80mm]{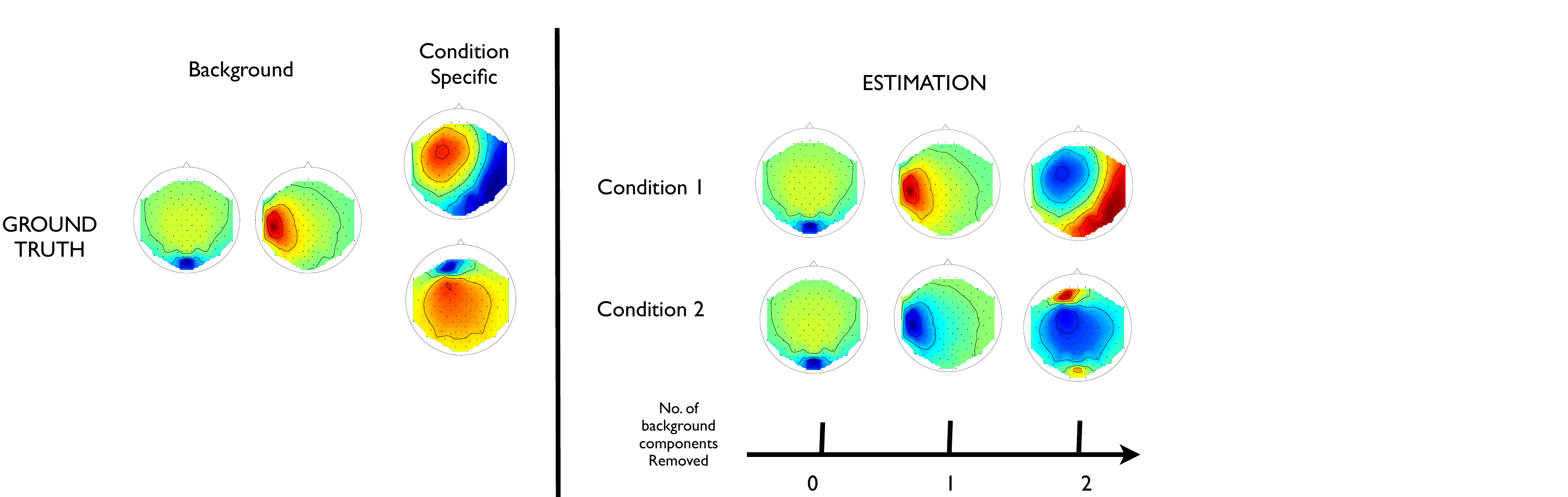}\\
\end{array}$
\end{center}
\caption{The figure displays the results of the forward model simulation when the number of conditions is two and the background non-stationarities are 
stronger than the condition specific non-stationarities. Here, using SSA, one obtains the background stationarities as estimated "condition-specific" non-stationarities
which agree across conditions after zero and one background non-stationarities have been removed. The condition specific non-stationarities are then 
obtained after the removal of two background non-stationary directions.
\label{fig:ForwardModelTwoClasses}}
\end{figure}

\begin{figure}[h]
\begin{center}$
\begin{array}{cc}
\includegraphics[width=110mm,clip=true,trim=3mm 85mm 20mm 80mm]{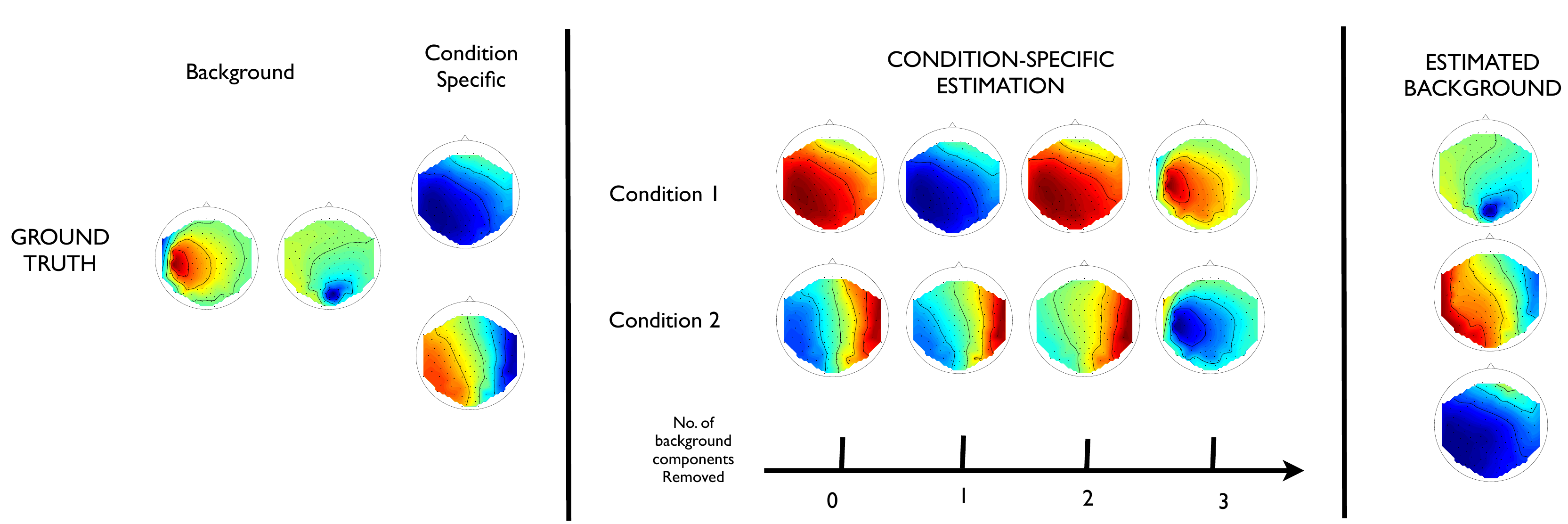}\\
\end{array}$
\end{center}
\caption{The figure displays the results of the forward model simulation when the number of conditions is two and the condition specific non-stationarities are 
stronger than the background non-stationarities. Here, after removal of zero, one and two background non-stationary directions, SSA finds \emph{different}
non-stationary patterns on each condition. However, after removal of three background non-stationary directions, SSA finds non-stationary patterns on each condition
which agree, and arise, in fact, from the true background non-stationarities. The reason that the condition specific non-stationarities are similar after removal of
zero and one background non-stationary directions, is because the first estimated \emph{background} non-stationary direction is nevertheless correct, 
the EMG component, so the condition specific most non-stationary directions are preserved. These are then removed as the estimated second and third most 
non-stationary background directions.
\label{fig:ForwardModelSpike}}
\end{figure}

\subsection{Application to Change-Point Detection}
\label{sec:cp_d}

In this section, we demonstrate that the proposed method can be used to find \textit{relevant} change-points 
in the data from a particular condition of interest, even when there are \textit{irrelevant} 
distribution changes that are shared by all conditions. In reality, this situation occurs e.g.\ when 
there are artifacts in the data that are independent of the experimental condition. 

Change point detection refers to the task of detecting points in time in which the distribution 
of a time series changes from one state to another; this task occurs potentially in sleep staging of 
EEG recordings \cite{switch_dynamics} and in early warning systems for epilepsy patients using chronically 
implanted electrodes \cite{Celka20021}, among other applications \cite{Speech1,Econ1,Econ2}. There
exists a wide range of different methods for change-point detection (also called time series segmentation). 
In our simulations, we use single linkage clustering based on a symmetrized Kullback-Leibler divergence 
measure~\cite{Gower69SingleLinkage,blybunmeimul12feature_DISCRIM} between segments of the time series.

\begin{figure}
\begin{centering}
\includegraphics[width= 130mm,clip=true,trim=30mm 160mm 10mm 15mm]{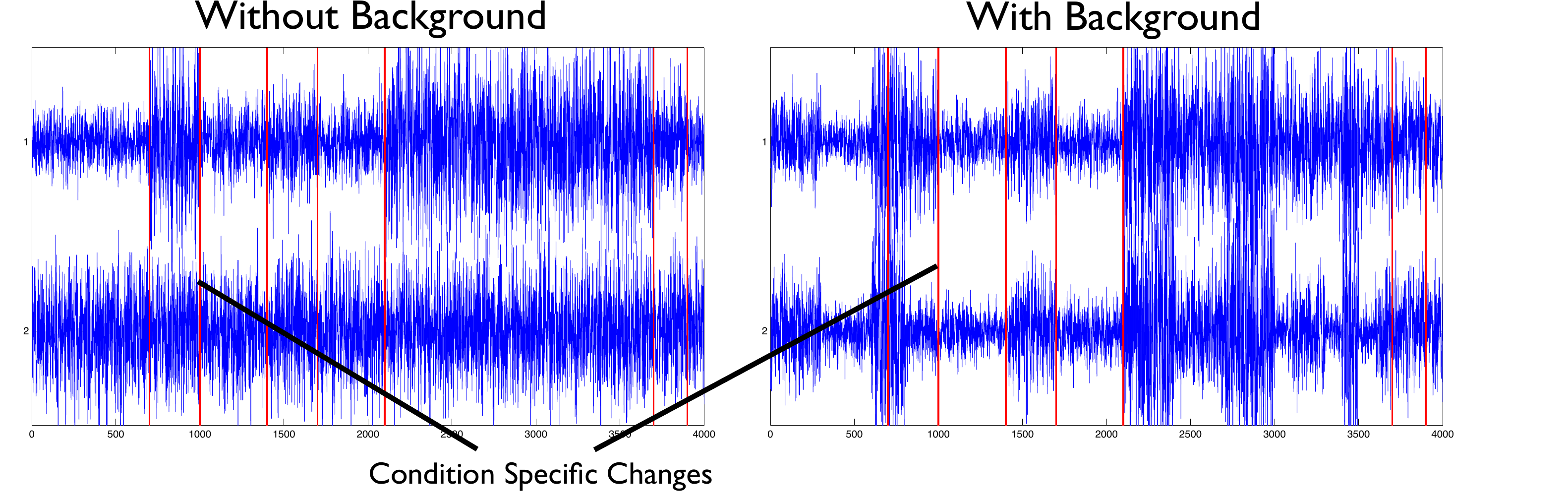}
\caption{The change point detection task is made more difficult when the change points of interest (red lines) in the condition specific system (left) are masked by background change points (right). 
The left panel displays only the condition specific system $B_k\mathbf{r_k}(t)$ whereas the right hand panel displays the entire data corresponding to $A \mathbf{s}(t) + B_k\mathbf{r_k}(t)$.
\label{fig:TempSegIllusBack}}
\end{centering}
\end{figure}

\begin{figure}[ht]
\begin{center}
\includegraphics[width = 120mm,clip=true,trim=2mm 45mm 40mm 30mm]{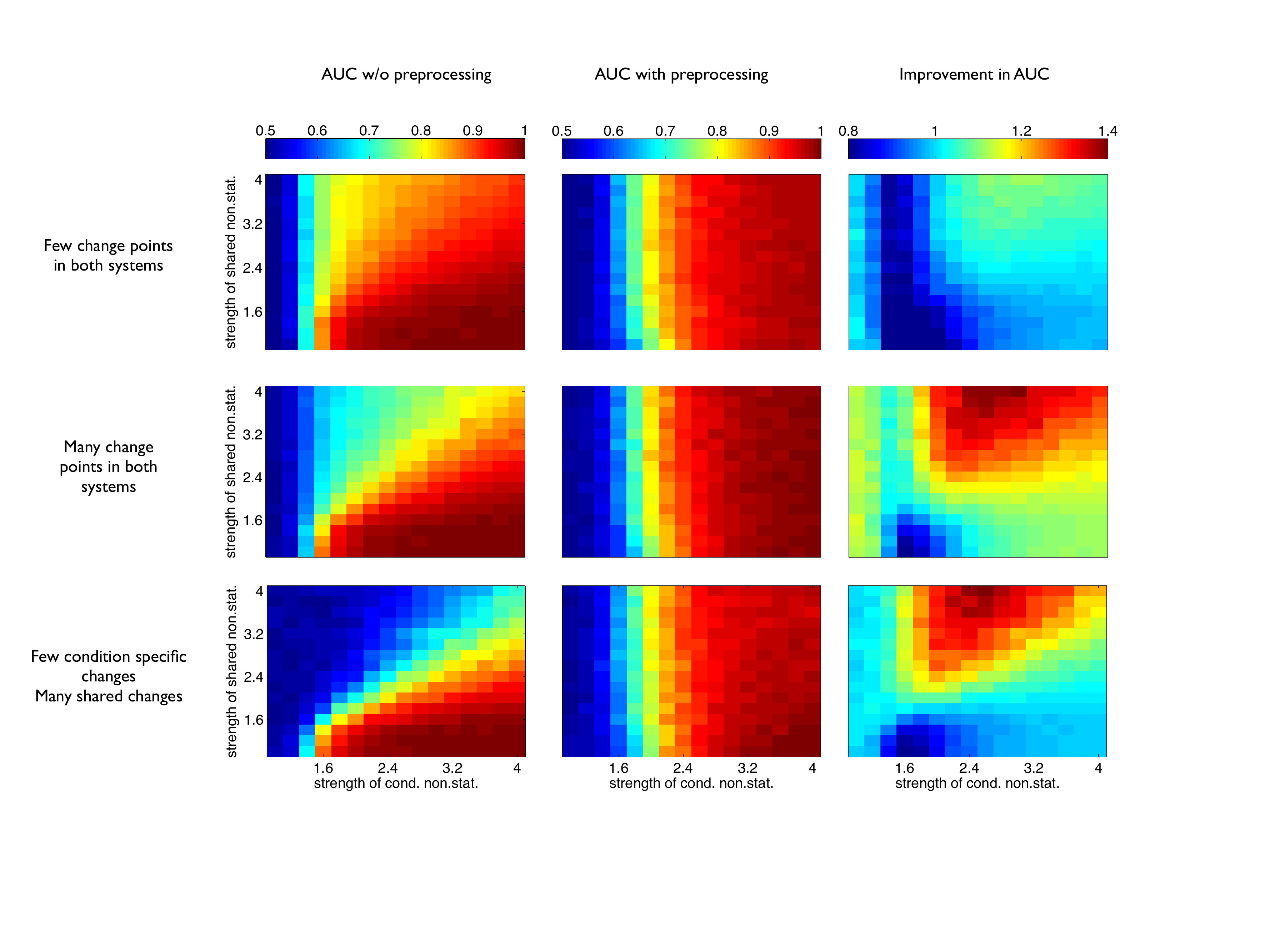}
\caption{
Results of the pre-processing for change point detection. The rows of the $3 \times 3$ grid 
correspond to three different scenarios: few background change points, many background change-points, and
few condition-specific changes but many background changes. From left to right, the columns 
shown the AUC without pre-processing (baseline), the AUC with pre-processing and a comparison
between the two (improvement). In each panel, we vary the strength of the background 
non-stationarity along the horizontal axis ($\nu_\text{shared}$), and the strength of the condition-specific 
non-stationarity along the vertical axis ($\nu_\text{specific}$). The precise parameters for transition are in the top row $p_\text{shared} = 0.8, p_\text{specific} = 0.8$, in the middle row $p_\text{shared} = 0.6, p_\text{specific} = 0.6$ and in the bottom row $p_\text{shared} = 0.1, p_\text{specific} = 0.8$.
The middle column pictures are similar because once the correct background non-stationary space has been correctly found and the background non-stationarities projected out (which is possible in all three cases) then the difficulty of the task is equal in all three simulations on this space.
\label{fig:temp_seg_class}}
\end{center}
\end{figure}

Here, we test the hypothesis that the performance of 
change-point detection can benefit from removing non-stationary directions (feature extraction) that 
are not specific to the condition of interest. Thus we suggest a \emph{pre-processing} step which 
aims to boost the performance of change-point detection, when the model applies.
Previous work 
\cite{blybunmeimul12feature_DISCRIM} has shown that removing stationary directions in a pre-processing step 
improves the performance of change-point detection. Here, we go one step further and distinguish between 
relevant and irrelevant non-stationary changes by making use of data from other conditions 
during feature extraction to remove the background \emph{non-stationary changes}.The problem 
setting is illustrated in Figure~\ref{fig:TempSegIllusBack}; here we see that infrequent and weak condition specific
changes are masked by stronger and more frequent background changes. This renders difficult detecting
the condition specific changes reliably. The solution we propose is to project the data to the background stationary
sources.

The synthetic data is generated as described in Section~\ref{sec:data}. The total number of dimensions
is 10, the number of conditions is 10 and the number of non-stationary directions is 2 and 1
for the background and the condition-specific systems. In the simulations, we vary
the strength of the non-stationarity in both systems ($\nu_\text{shared}$ and $\nu_\text{specific}$), 
the likelihood of the relevant change points in the condition of interest ($p_\text{specific}$), and the likelihood 
of irrelevant change points in the background system ($p_\text{shared}$). 

As a baseline, we directly apply the change-point detection algorithm to the data from the condition 
of interest and discard the remaining nine data sets. We compare this approach against the augmentation
of this algorithm by our pre-processing (feature extraction). More specifically, we apply SSA to 
estimate the nine most stationary components \textit{over all conditions} and then apply the change-point 
detection algorithm to these components on the data from the condition of interest. We measure the performance 
in terms of the area under the receiver-operator curve (AUC).

The results are shown in Figure~\ref{fig:temp_seg_class}. The right column of the $3 \times 3$ grid 
shows the improvement in AUC due to our pre-processing for three different scenarios. When there 
are few background irrelevant change points (top row) \emph{and} few condition specific change points, the performance of change-point detection is 
relatively unaffected, there is improvement using a preprocessing step only when the condition shared changes are pronounced enough
to yield reasonable estimation of the desired projection. 

In the second row the case in which there are many irrelevant background change points as well as many condition-specific change points of interest.
which interfere with the detection of the relevant change points in the conditions of interest.
Thus, in the middle panel of the right column, we see that the proposed pre-processing leads to a 
significant improvement in performance for almost all cases. However, when both systems have
low non-stationarity (bottom left corner of the panel), 
it is hard to identify the informative 
background stationary components; hence the performance does not increase. 

The last simulations whose results are displayed in the bottom row, confirm that it is the frequency
of the changes in the distracting condition \emph{shared} system which determine whether a high 
level of improvement is to be expected. Here, as in the second row, the frequency of background
changes is high and the frequency of the desired condition-\emph{specific} changes is low.
As in the second row, improvement is to be expected in all cases except the cases in which
the strength of the background changes are not sufficient to yield reliable estimation
of the shared stationary projection.

%
%
%
%

\section{Application to Brain Computer Interfacing Data}
\label{sec:application}
A \emph{brain computer interface} is a device which allows a human
subject to perform a limited number of commands on a computer using
only the activity observable using a brain imaging device, such as an
fMRI scanner, NIRS imaging system or most commonly an EEG
(electroencephalography) cap. In the motor imagery paradigm for EEG
\cite{oai:eprints.pascal-network.org:3318,BlankertzNIMG2007}, machine
learning algorithms may be trained on data recorded during
imaginations of hand movements corresponding to the desired direction
of movement of a cursor on a computer screen. One practical problem
with such a device is, however, that the distribution over scalp
voltages in electroencephalography (EEG) recordings is highly
non-stationary. This has the effect that Brain Computer Interfaces
based on machine learning degrade in their performance over time
\cite{oai:eprints.pascal-network.org:5102,oai:eprints.pascal-network.org:3317,oai:biomedcentral.com:1471-2202-10-S1-P85,carmenNeuralComp}. 
To deal with this problem, to maintain high levels of performance,
numerous machine learning methodologies have been proposed, based on unsupervised adaptation \cite{oai:eprints.pascal-network.org:5102}, robustification \cite{oai:eprints.pascal-network.org:3317},
penalization of non-stationarity \cite{conf/icassp/WojcikiewiczVK11,sCSP} and covariate shift adapation \cite{Sugiyama:2007:CSA:1314498.1390324}.
An
interesting \emph{neuroscientific} question, however, is to diagnose the physical source of
this non-stationarity.  On the one hand it may be the case that
non-stationarity results to a high degree from artifactual sources,
such as EMG activity, ocular activity or disturbances resulting from
loose electrodes \cite{conf/icassp/WojcikiewiczVK11,sCSP, GroupWiseSSA}.  On
the other hand, the fact remains that the activity in the EEG
recording resulting from neural activity is also highly
non-stationary, albeit to a lesser degree and often over a different
timescale than the non-stationarity originating in artifacts.  The
targeted study of each of these non-stationary contributions has been
until recently been confounded by the fact that only a linear mixture
of their respective activities may be observed at the scalp
electrodes. In particular, the contribution of the neural component to
the overall non-stationarity is typically weaker than the contribution
of the artifactual activity, making, in particular, the targeted study
of neural non-stationarity especially problematic.  Nevertheless, the
model studied in the present contribution presents the possibility of
prying apart neural and artifactual non-stationarity. Given that this
may be achieved, there exists the prospect of better understanding the
neural changes which occur during EEG based BCI use.

To illustrate this possibility we display scalp maps in Figure~\ref{fig:bci_changes} of a BCI subject. 
For each condition (left imagination and foot imagination) the most non-stationary pattern in that condition
\emph{subsequent} to the removal of a number, $m$, of shared non-stationary directions has been visualized. The figure
displays these patterns for $m = 0,1,\dots,6$.
We observe in the results
that the model studied above is applicable to the data set at hand: namely, there is a common
non-stationary component to both of the condition systems. When no or few class-shared non-stationarities are removed,
then the patterns are similar between the two conditions.  The structure of the displayed patterns indicate that the activity generating
this non-stationarity corresponds to an EMG artifact.
We see, however,
by increasing the number of common non-stationary sources removed, that 
differences in the patterns between the estimated condition non-stationarities emerge. 
For example, the patterns in the left hand condition, after removal of six background non-stationary sources, display high weights contralateral to the imagined hand and are smooth in their topography, in contrast to the 
most non-stationary patterns. These facts indicate that the patterns may contain information regarding \emph{task relevant neural} changes
observable in the data from each condition. 

We now extend our analysis to include the data from all 80 subjects recorded for the study at hand \cite{NeuroPred,PhysPred}.
In order to illustrate the usefulness of the proposed model for a general BCI motor imagery subject, we use the following heuristic
to quantify to what extent the model is suitable for any given dataset:

\begin{enumerate}
\item We compute a similarity measure between each of the 2
condition specific most non-stationary patterns for every choice of the dimensionality of background non-stationary subspace, $d_\s$, including 
the case where there are assumed to be \emph{no} background non-stationary sources. An example of the results of such a computation 
are displayed in Figure~\ref{fig:NonStatSpectrum}
\item We then compute the \emph{difference} in the average similarity of the condition specific non-stationary patterns for $d_\s=5, \dots, 10$ and the similarity
of the condition specific patterns for $d_\s = 0$.
\item If this difference is \emph{greater} than a predefined threshold, then we conclude that the model is applicable for this dataset.
\end{enumerate}

An important point to note here, which is important for the validity of the heuristic, is that in the null case in which the model does not apply, then the
similarity between the condition specific patterns for $d_\s = 0$ is on average smaller than the similarity between the condition specific patterns for 
$d_\s = 5,\dots,10$.
The results of applying this heuristic to all 80 BCI motor imagery subjects is displayed in Figure~\ref{fig:EffectAcrossSubjects}.
They show for a significant fraction of subjects there are clear background non-stationary components; this may be concluded
from the fact that the most non-stationary components estimated on the conditions separately are identical, up to estimation
error and noise. This implies that the model we propose in this paper is highly appropriate for studying the changes in distribution of the EEG 
of BCI subjects. 

Note at this point that this heuristic is designed to indicate the usefulness of the model for the dataset. A more principled approach 
must wait for further work due to the advanced statistical exposition and innovation such an approach would imply. For example, an approach which compares the similarity 
of vectors over different dimensionalities must correct for this discrepancy, see, for example \cite{PrincipleAngle} for theorems facilitating this correction. 

One clear limitation of the dataset is that only two conditions are available. Finding the background non-stationarities
is thus only possible because the condition specific non-stationarities are weaker than the background non-stationarities, for most subjects in this dataset.
Occasionally the condition specific non-stationarities are not weaker, see Figure~\ref{fig:EffectAcrossSubjects} and Figure~\ref{fig:subject_selection}. 
Here we observe that a strong condition specific non-stationary component is present. For such a subject, the presence
of more experimental conditions would guarantee that this component is nevertheless not neglected.
However, despite this, we see that in a the majority of cases the condition specific components are \emph{not} stronger than the background components and
the sucessful removal of the background components is possible, and, thus that the model is suitable for this type of dataset. 

\begin{figure}
\begin{centering}
\includegraphics[width=120mm]{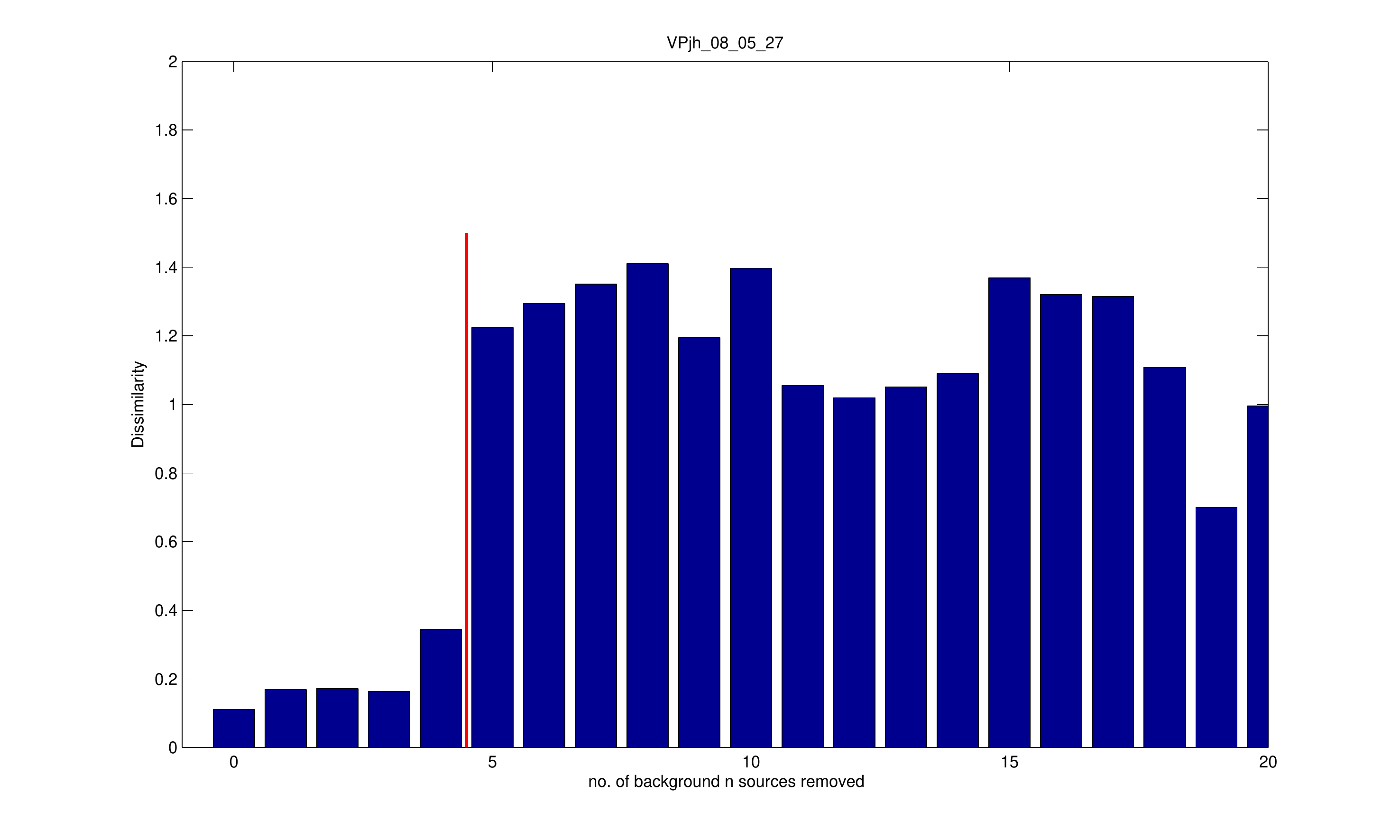}
\caption{The figure displays the results for a single subject (VPjh) of computing Step 1 of the heuristic for determining applicability of the model for a given dataset. 
The $x$-axis displays the number of background sources removed and the $y$-axis displays the dissimilarity (1-norm between normalized patterns) between the condition specific non-stationary patterns on each of the two classes.
The red line indicates the transition from a non-random level of dissimilarity to a random level, suggesting that the number of background non-stationary sources is five.
 \label{fig:NonStatSpectrum}}
\end{centering}
\end{figure}

\begin{figure}
\includegraphics[width=120mm,clip=true,trim=30mm 25mm 25mm 40mm]{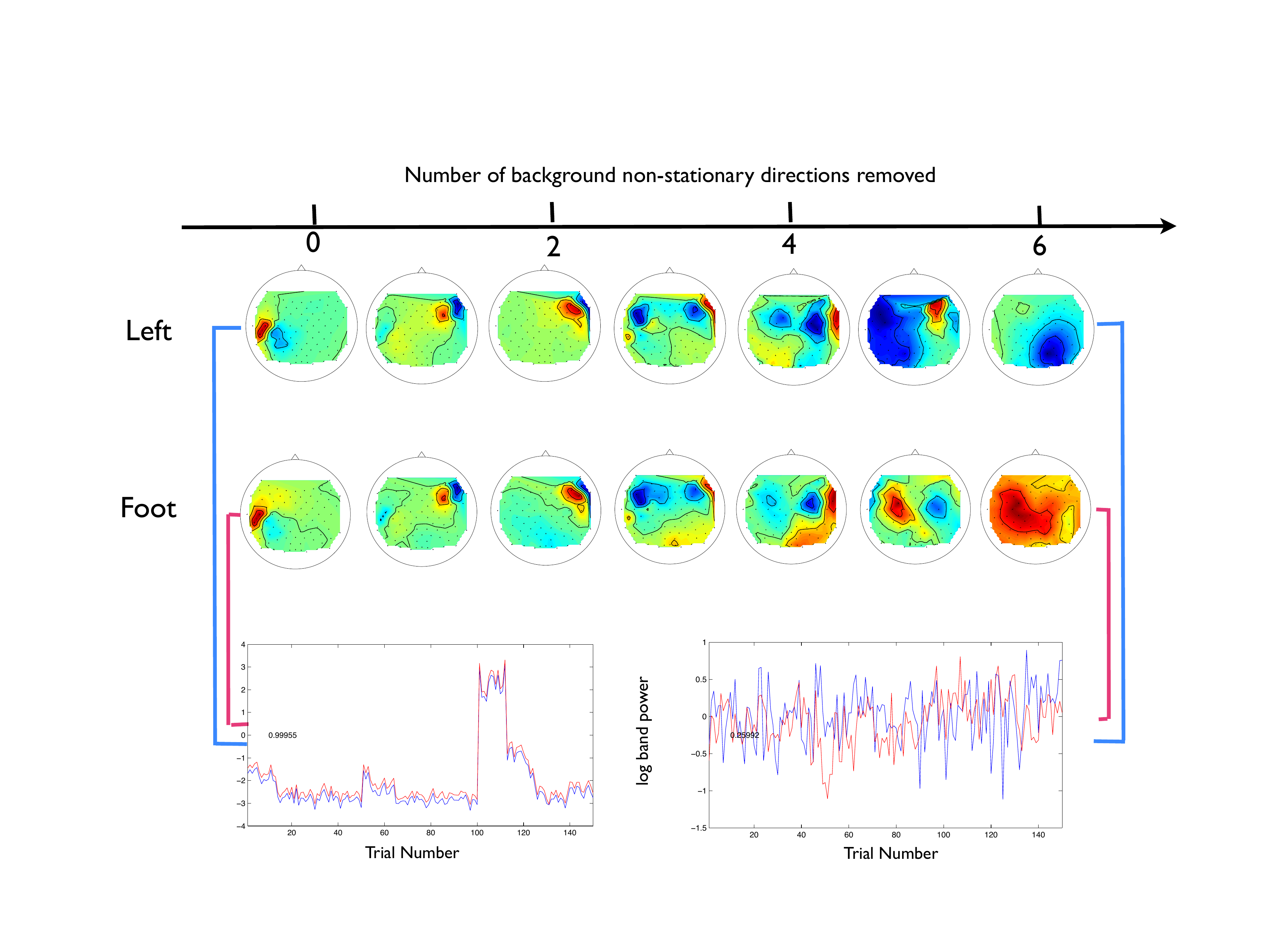}
\caption{The figure displays scalp maps from a single BCI subject. The upper row displays maps from the first condition, whereas the second row from the second condition.
The conditions are ordered in terms of how many jointly non-stationary components were removed. The numbers removed are, from left to right, 0,1,...,6.
The difference between the conditions becomes clearly visible only after 5 directions have been removed. The similarity between the preceding components
is clearly visible, between conditions. The traces at the
bottom of the figure, colored red and blue corresponding to the condition, display the log band power recorded on each trial over the course of the experiment. The highest components
are highly correlated, the lower, uncorrelated. These band-powers are calculated on the data $ Q_k^n P^\s U \mathbf{x_k}(t)$.
 \label{fig:bci_changes}}
\end{figure}

\begin{figure}[h]
\begin{center}
$\begin{array}{c} \includegraphics[width=100mm,clip=true,trim=3mm 86mm 10mm 90mm]{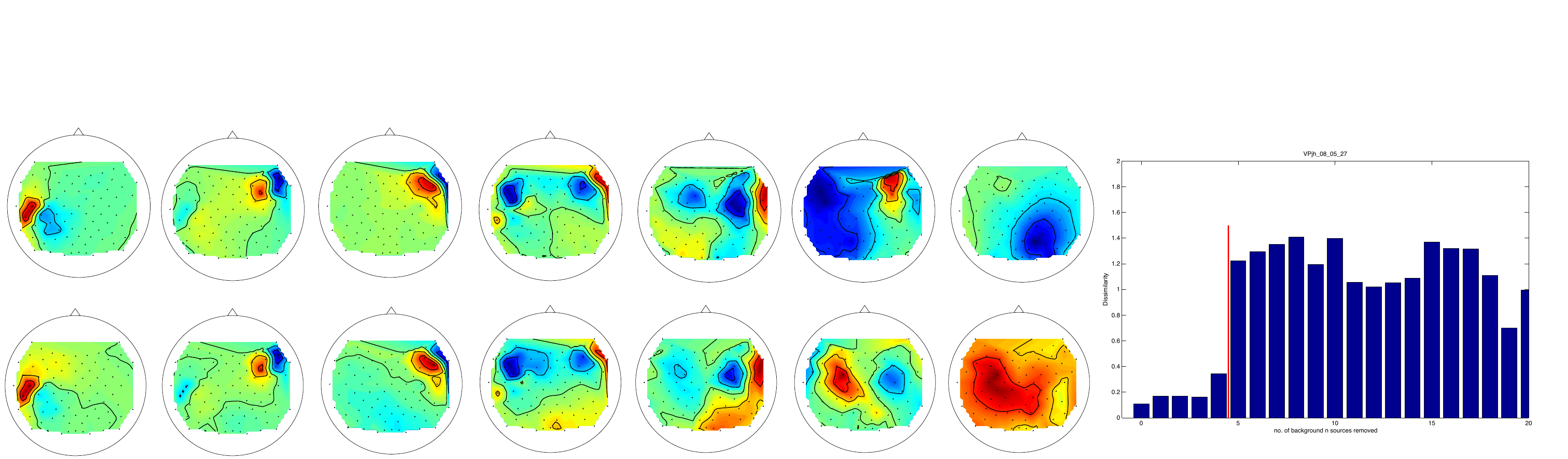} \\
 \includegraphics[width=110mm,clip=true,trim=20mm 101mm 10mm 90mm]{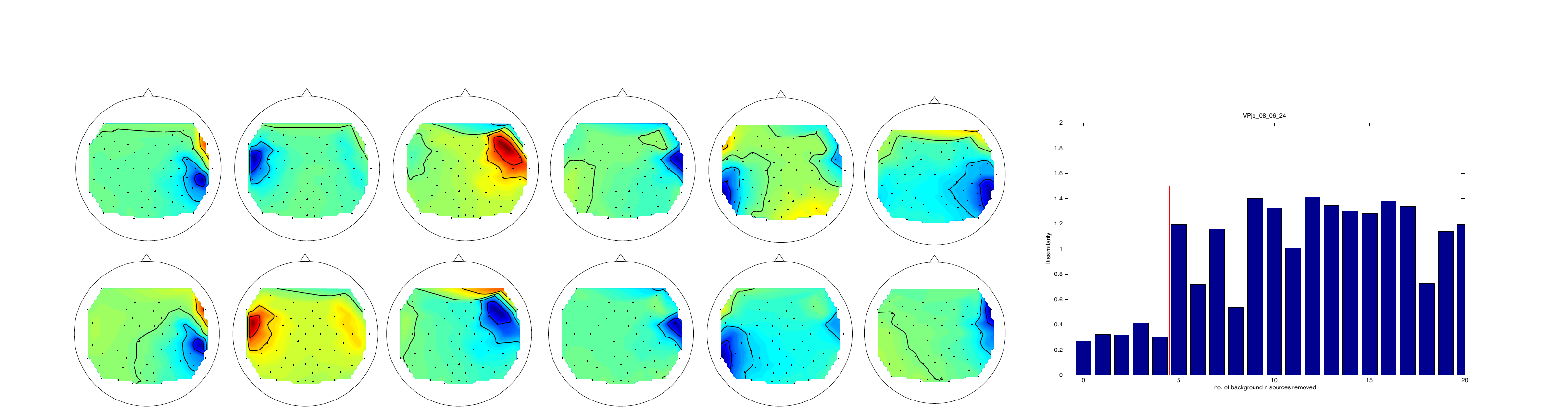} \\
  \includegraphics[width=102mm,clip=true,trim=58mm 95mm 10mm 94mm]{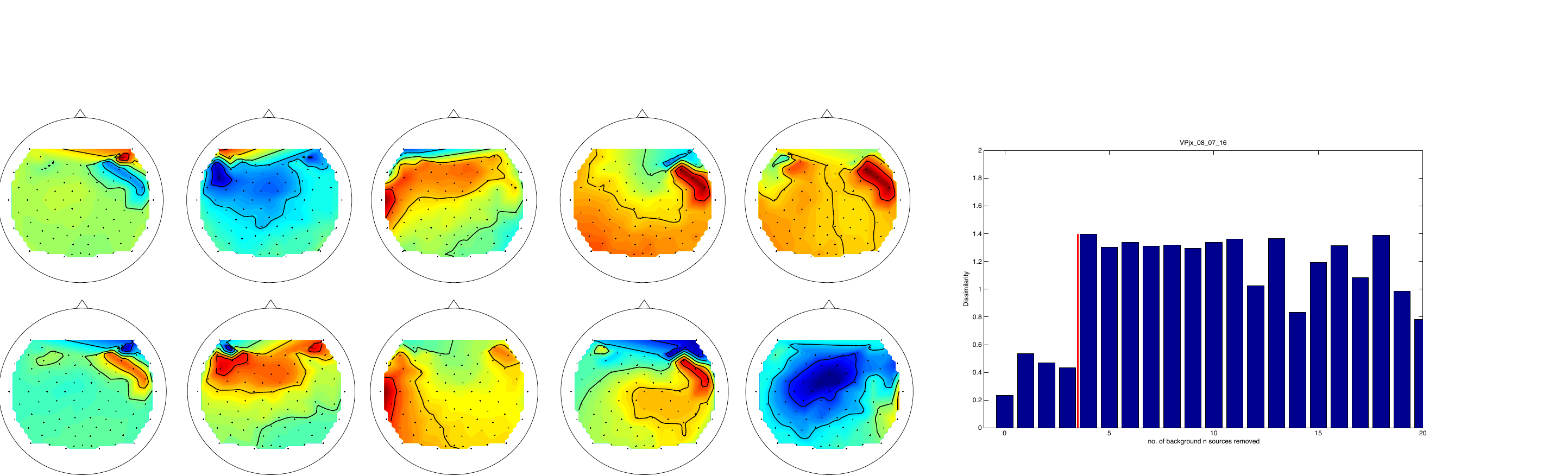} \\
\includegraphics[width=100mm,clip=true,trim = 30mm 82mm 10mm 105mm]{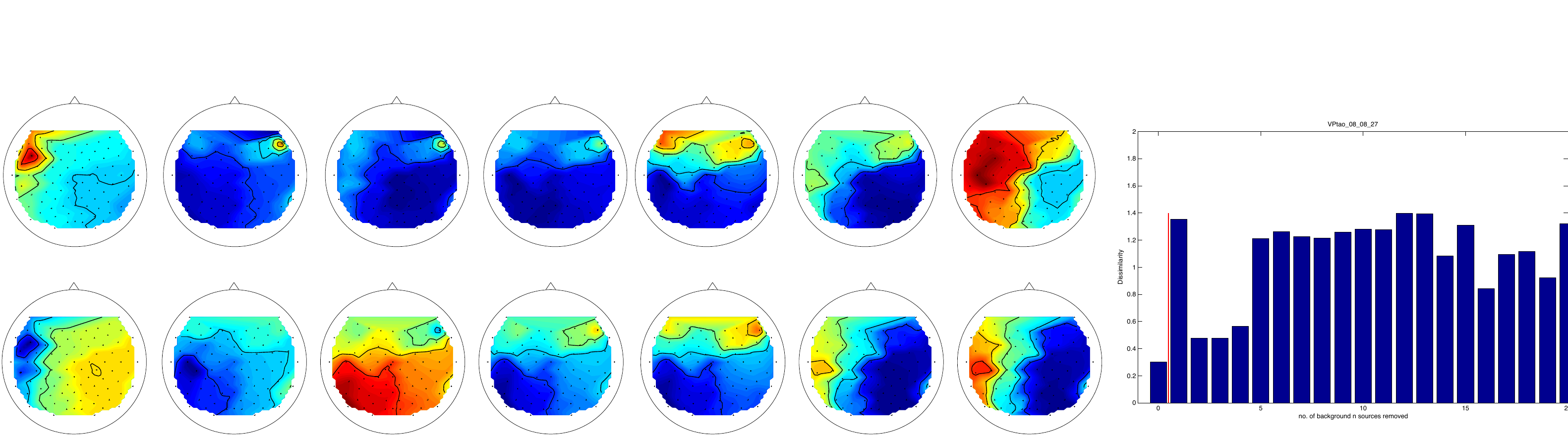}
\end{array}$
\end{center}
\caption{The figure diplays three subjects (the top three, VPjh, VPjo, VPjx) with clear background non-stationary directions. The scalp maps are displayed on the left and the 
similarity between components on the right as per Figure~\ref{fig:NonStatSpectrum}. The fourth and final subject (bottom, VPtao) displays background  non-stationary directions
but also a highly non-stationary condition specific component. The component is located frontal right in the first condition, and has a localized topography. We conjecture 
this is caused by an EMG artifact which is biased towards the first condition. 
The red line on the right displays the number of background non-stationary components estimated by the heuristic. This choice is visible in the similarity of the patterns displayed on the right; in the uppermost row,
we display an extra pair of patterns to display the differentiation and class structure in the condition specific patterns. In the bottom row we display additional patterns to show that, in this case, the heuristic's choice is conservative.
\label{fig:subject_selection}}
\end{figure}

\begin{figure}
\begin{centering}
\includegraphics[width=120mm]{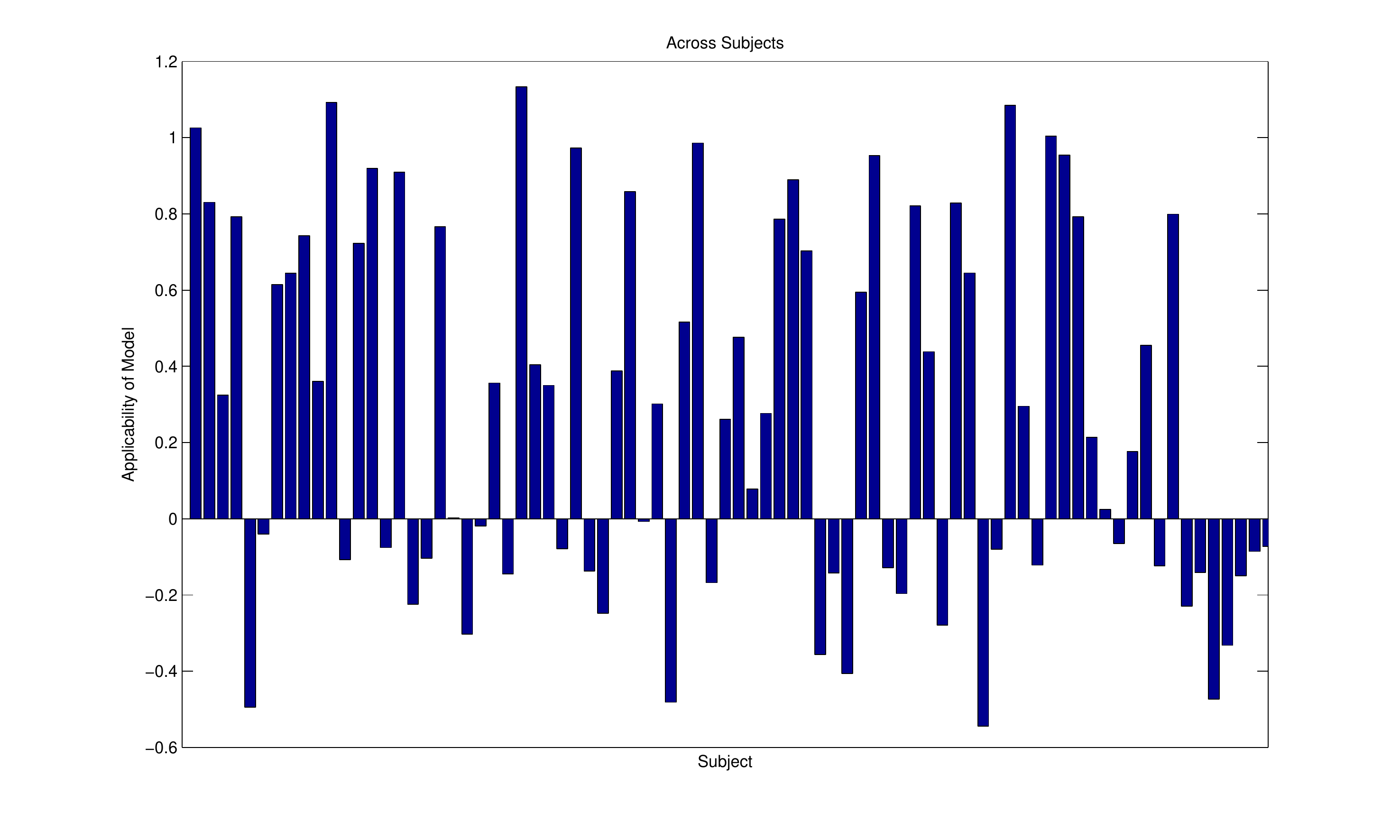}
\end{centering}
\caption{The figure displays the results of applying the heuristic for estimating the applicability of the model. The $y$-axis shows the difference in the average similarity of the condition specific non-stationary patterns for $d_\s=5, \dots, 10$ and the similarity
of the condition specific patterns for $d_\s = 0$. The $x$-axis displays the subject number. The figure shows that for a large proportion of the subjects in the dataset, the model is applicable.
\label{fig:EffectAcrossSubjects}}
\end{figure}

%
%
%
%

\section{Concluding Discussion}
\label{sec:conc}
Neural signals are subject to change which is typically induced by the
experimental paradigm. In addition, typically a number of artifactual
signals are present in an experiment that hamper a possible
interpretation: they may consist of drifts, or also other more complex
factors that are common between behavioural conditions. This work has
contributed with its explorative framework which enables the removal of these
uninteresting non-stationary components. Thus, for the first time
condition specific non-stationarities can be extracted and
physiologically interpreted. In particular this was demonstrated on simulated EEG examples from
a realistic forward model as well as on real experimental data from a
large scale BCI study.

An important assumption that needs to be made is that the condition
specific non-stationary patterns $B_k^\n, B_{k'}^n$ should
be independent for $k \neq k'$. If these are highly dependent, then no
estimation of $P^\s$, the projection to the background stationary
sources, is possible. If independence does not hold then 
the condition-specific features of interest may significantly overlap in
their projection with the condition shared non-stationarities, and
information may be lost. The requirement for independence is, however, an intrinsic feature of
the model and not a deficiency of the estimation method. If independence isn't assumed in \emph{some}
form then estimation by any method is impossible.
However, as discussed in Section~\ref{sec:con_proof},
the assumption  of independence of origin is realistic, in particular since the
non-stationarities specific to each condition result from separate
underlying physical systems. Note that it will often be possible to know
a priori if these assumptions are likely to hold for a particular application.
For example, in BCI, different brain regions govern the control of separate motor
imaginations and thus the condition specific non-stationarity patterns should be 
dissimilar from one another. If the assumption of independence of origin does not hold, then, estimation may still be possible under a slightly difference independence assumption.
This is because it may be plausibly assumed that 
the evolution of non-stationarity in the \emph{time-domain} is independent in each condition-specific system.
This should also guarantee consistency; we postpone the proof of this result to a later publication.

Another important assumption we make is that no
condition should be allowed to contribute disproportionately to the
non-stationarity in the overall data. As we argued earlier, this
assumption is realistic in applications, in which the background
non-stationarity, for example, artifacts in EEG recordings, are
stronger than the condition specific non-stationarities. Moreover, this assumption is related to assuming that
the background non-stationarities are stronger than the condition-specific non-stationarities. If it is the case that one condition-specific system 
exhibits high non-stationarity, then the confounding influence of the background non-stationarity is reduced for this class in any case.
Thus making special allowance for background non-stationarity in such a case is less pressing.

After \emph{removing} the background non-stationary activity
from a data set, further analysis must be confined to a subspace of
the orthogonal complement in data space of the removed directions. 
However, as we demonstrated in our change-point detection analysis,
this subspace may yield important performance gains over analysis
on the entire space.

In summary, we have shown that the model and methods we propose 
boosts the performance of change point detection and allows the data analyst to
discover condition specific changing components in BCI data. 
Future studies will apply the novel framework for analysing
experimental data to primate experiments that study learning and
plasticity. In addition, we will study condition specific
non-stationarities in multimodal data \cite{biessmannreview} and in
co-adaptive BCI \cite{carmenNeuralComp}. It should also prove interesting to
localize the respective sources of different non-stationary nature
\cite{HauNikZieMueNol08, HauTomDicSanBlaNolMue10}.

\section*{Acknowledgements}
The authors would like to thank Daniel Bartz, Franz Kiraly and Stefan Haufe for their valuable comments on the topic of this paper. 
D.A.J.B was supported in part by a grant from the research training group GRK 1589/1 "Sensory Computation in Neural Systems" and by the Bernstein Focus Neuro-technology (BFNT), funded by the BMBF, FKZ 01GQ0850; 
F.C.M. was supported by BMBF and DFG in the German-Japanese research project "Robust Adaptive BCI in non-stationary environments" FKZ 01GQ1115.
; K.-R.M was supported in part by the World Class University Program through
the National Research Foundation of Korea funded by the Ministry of
Education, Science, and Technology, under Grant R31-10008. 

\newpage
\appendix
\appendixpage
\section{Proof of Proposition 1}
\label{app:proofs}
\begin{proof}

The proof strategy is as follows: if we show that the statement holds for each of the non-stationary projections $Q^\n_k$ simultaneously, where $Q_k = (B_k)^{-1}$, then clearly the statement holds for a random projection $P$.
We need to show that the projection of the system by $Q_k^\n$ is stationary, for all $K$, in the limit of $K$ and $D$.

For any epoch indexed by $i$, any class index $k'$:

\begin{eqnarray*}
\vectornorm{\frac{1}{K} \sum_{k=1}^K Q_{k'}^\n(R_i^k - \bar{R}_k)Q_{k'}^\n} &\leq& \frac{1}{K}\vectornorm{R_i^{k'} - \bar{R}_{k'}} + \frac{1}{K} \sum_{k \neq k'}  \vectornorm{Q_{k'}^\n(R_i^k - \bar{R}_k)Q_{k'}^\n} \\
&\leq& \frac{1}{K}\Gamma + \frac{\kappa}{K}\sum_{k\neq k'} \vectornorm{S^\top I_{d_n} S}
\end{eqnarray*}

The right hand quantity in the first line is bounded above by the quantity on the right in the second line, $\frac{\kappa}{K}\sum_{k\neq k'} \vectornorm{S^\top I_{d_n} S}$ where $S$ is a $d_n$ dimensional random projection matrix, $I_{d_n}$ is a matrix of zeros apart from the left, top hand $d_n$ rows and columns which are equal to the identity matrix and $\kappa$ is a constant.
The random variable on the right under the sum, $ \vectornorm{S^\top I_{d_n} S}$ has mean $\mu_D$ of order $\frac{1}{D}$ and standard deviation $\sigma_D$ of order $\frac{1}{\sqrt{D}}$.
This is because as $D$ grows, for fixed $d$, the rows of the matrix $P$ tend to independent samples from the distribution over one dimensional subspaces. The dot product between random one dimensional subspaces has density of order $\sqrt{D}x^{-1}(1-x^2)^{(D-3)/2}$  with moments of order $1/\sqrt{D}$
 \cite{PrincipleAngle}.
Each row in the matrix $P^\top I_{d_n}$, is a dot product, between basis vectors of unit length of random one-dimensional subspaces. Thus, each row has norm bounded by constant order. Moreover the effect of the truncated matrix in the second term is to truncate the entries in $P$, so that each entry of the entire matrix $P^\top I_{d_n} P$ may be approximated as dot products between random subspaces and the entries may be assumed independent. Thus the 
 mean of each term in the Froebenius norm is at most of the order of the variance of the dot product between spanning vectors of random one dimensional subspaces and the standard deviation at most of order of the square root of the fourth moment, which both decrease as $1/\sqrt{D}$ and $1/(D^{1/4})$.
Thus the sum term is distributed in the limit of $K$ with standard deviation bounded by $\frac{\beta}{\sqrt{K}D^{1/4}}$ and mean by $\frac{q}{\sqrt{D}}$ for some constants $q$ and $\beta$. 

Thus by Chebyshev's inequality:

%
%

\begin{equation*}
\text{Pr}(\frac{1}{K} \sum_{k\neq k'}  \vectornorm{S^\top I_{d_n} S} > \epsilon/\kappa + \frac{q}{\sqrt{D}}) <  \frac{\kappa^2\beta^2}{\epsilon^2 K \sqrt{D}}
\end{equation*}

Which implies,

\begin{equation*}
\text{Pr}( \frac{1}{K}\Gamma + \frac{\kappa}{K}\sum_{k\neq k'} \vectornorm{S^\top I_{d_n} S}> \frac{\kappa}{K}\Gamma + \epsilon + \frac{c}{\sqrt{D}}) < \frac{\kappa^2\beta^2}{\epsilon^2 K \sqrt{D}}
\end{equation*}

Then by Boole's inequality the probability that each one of these inequalities holds gives that
for any $k$:

\begin{equation*}
\text{Pr}( \frac{1}{K}\Gamma + \frac{\kappa}{K}\sum_{k\neq k'} \vectornorm{S^\top I_{d_n} S}> \frac{1}{K}\Gamma + \epsilon + \frac{c}{\sqrt{D}}) < \frac{\kappa^2\beta^2}{\epsilon^2 \sqrt{D}}
\end{equation*}

Putting all of this together yields the estimate:

\begin{equation}
\text{Pr}\left(\vectornorm{\frac{1}{K} \sum_{k=1}^K Q_{k'}^1(R_i^k - \bar{R}_k)Q_{k'}^1} > \mathcal{O}(\frac{1}{K}) + \mathcal{O}(\frac{1}{\sqrt{D}}) + \epsilon\right) <  \frac{\kappa^2\beta^2}{\epsilon^2 \sqrt{D}}
\end{equation}

Which is:

\begin{equation}
\text{Pr}\left(\vectornorm{\frac{1}{K} \sum_{k=1}^K Q_{k'}^1(R_i^k - \bar{R}_k)Q_{k'}^1} > \mathcal{O}(\frac{1}{K}) + \mathcal{O}(\frac{1}{\sqrt{D}}) + \mathcal{O}(\frac{1}{\sqrt{\delta \sqrt{D}}})\right) <  \delta
\end{equation}

\end{proof}

\bibliographystyle{plain}
{\small
\bibliography{../bib/Testing,../bib/stable,../bib/ssa_pubs,../bib/segmentation,../bib/class_discrim,../bib/supervised,../bib/LFP}}

\begin{thebibliography}{10}

\bibitem{PrincipleAngle}
P.A. Absil, A.~Edelman, and P.~Koev.
\newblock On the largest principal angle between random subspaces.
\newblock {\em Linear Algebra and its Applications}, 414:288--294, 2006.

\bibitem{Speech1}
R.~Andre-Obrecht.
\newblock A new statistical approach for the automatic segmentation of
  continous speech signals.
\newblock {\em IEEE Trans.~Acoustics, Speech, Signal Processing}, 36(1):29--40,
  1988.

\bibitem{biessmannreview}
F.~Bie{\ss}mann, S.M. Plis, F.C. Meinecke, T.~Eichele, and M\"uller K.R.
\newblock Analysis of multimodal neuroimaging data.
\newblock {\em IEEE Reviews in Biomedical Engineering}, 4:26--58, 2011.

\bibitem{BlankertzNIMG2007}
B.~Blankertz, G.~Dornhege, M.~Krauledat, K.-R. M\"uller, and G.~Curio.
\newblock The non-invasive {Berlin Brain-Computer Interface}: Fast acquisition
  of effective performance in untrained subjects.
\newblock {\em NeuroImage}, 37:539--550, 2007.

\bibitem{oai:eprints.pascal-network.org:3317}
B.~Blankertz, M.~Kawanabe, R.~Tomioka, F.~Hohlefeld, V.~Nikulin, and K.R.
  M{\"u}ller.
\newblock Invariant common spatial patterns: Alleviating nonstationarities in
  brain-computer interfacing.
\newblock {\em Advances in Neural Information Processing Systems}, 20:113--120,
  2008.

\bibitem{NeuroPred}
B.~Blankertz, C.~Sannelli, S.~Halder, E.M. Hammer, A.~K\"ubler, K.R. M\"uller,
  G.~Curio, and T.~Dickhaus.
\newblock Neurophysiological predictor of smr-based bci performance.
\newblock {\em Neuroimage}, 51:1303--1309, 2010.

\bibitem{oai:eprints.pascal-network.org:3318}
B.~Blankertz, R.~Tomioka, S.~Lemm, M.~Kawanabe, and K-R. M{\"u}ller.
\newblock Optimizing spatial filters for robust {EEG} single-trial analysis.
\newblock {\em IEEE Signal Processing Magazine}, 25(1):41 -- 56, 2008.

\bibitem{oai:biomedcentral.com:1471-2202-10-S1-P85}
B.~Blankertz and C.~Vidaurre.
\newblock Towards a cure for {BCI} illiteracy: machine learning based
  co-adaptive learning.
\newblock {\em BMC Neuroscience}, 10:85, July 2009.

\bibitem{blybunmeimul12feature_DISCRIM}
D.~A.~J. Blythe, P.~von B\"unau, F.C. Meinecke, and K.-R. M\"uller.
\newblock Feature extraction for change-point detection using stationary
  subspace analysis.
\newblock {\em IEEE Transactions on Neural Networks and Learning Systems},
  23(4):631--643, 2012.

\bibitem{Celka20021}
P.~Celka and P.~Colditz.
\newblock Time-varying statistical dimension analysis with application to
  newborn scalp {EEG} seizure signals.
\newblock {\em Medical Engineering and Physics}, 24(1):1 -- 8, 2002.

\bibitem{neuralVariabilityPremotor}
M.M. Churchland, B.M. Yu, S.I. Ryu, G.~Santhanam, and K.V. Shenoy.
\newblock Neural variability in premotor cortex provides a signature of motor
  preparation.
\newblock {\em Journal of Neuroscience}, 26:3697--3712, 2006.

\bibitem{Comon1994287}
P.~Comon.
\newblock Independent component analysis, a new concept?
\newblock {\em Signal Processing}, 36(3):287 -- 314, 1994.

\bibitem{Econ1}
H.~Cs\"org\"o and L.~Horv{\'a}rth.
\newblock Nonparametric methods for change point problems.
\newblock In P.R. Krishnaiah and C.R. Rao, editors, {\em Handbook of
  statistics}, volume~7, pages 403--425. Elsevier, New York, 2009.

\bibitem{Kopfmodel1}
VS~Fonov, AC~Evans, K~Botteron, CR~Almli, RC~McKinstry, DL~Collins, and the
  Brain Development Cooperative~Group.
\newblock Unbiased average age-appropriate atlases for pediatric studies.
\newblock {\em NeuroImage}, 54:313--327, January 2011.

\bibitem{Gower69SingleLinkage}
J.~C. Gower and G.~J.~S. Ross.
\newblock Minimum spanning trees and single linkage cluster analysis.
\newblock {\em Journal of the Royal Statistical Society}, 18(1):54--64, 1969.

\bibitem{PhysPred}
E.M. Hammer, S.~Halder, B.~Blankertz, C.~Sannelli, T.~Dickhaus, S.~Kleih, K.-R.
  M\"uller, and A.~K\"ubler.
\newblock Psychological predictors of {SMR-BCI} performance.
\newblock {\em Biological Psychology}, 89:80--86, 2012.

\bibitem{HauNikZieMueNol08}
S.~Haufe, V.~V. Nikulin, A.~Ziehe, K.-R. M\"uller, and G.~Nolte.
\newblock {{C}ombining sparsity and rotational invariance in
  {E}{E}{G}/{M}{E}{G} source reconstruction}.
\newblock {\em NeuroImage}, 42(2):726--738, Aug 2008.

\bibitem{HauTomDicSanBlaNolMue10}
S.~Haufe, R.~Tomioka, T.~Dickhaus, C.~Sannelli, B.~Blankertz, G.~Nolte, and
  K.-R. M\"uller.
\newblock Large-scale {EEG}/{MEG} source localization with spatial flexibility.
\newblock {\em NeuroImage}, 54:851--859, 2011.

\bibitem{switch_dynamics}
J.~Kohlmorgen, K.~R. M\"uller, J.~Rittweger, and K.~Pawelzik.
\newblock Identfication of nonstationary dynamics in physiological recordings.
\newblock {\em Biological Cybernetics}, 83:73--84, 2000.

\bibitem{Mandelblatt:2011q}
Y.~Mandelblat-Cerf, I.~Novick, R.~Paz, Y.~Link, S.~Freeman, and E.~Vaadia.
\newblock the neuronal basis of long term sensorimotor adaptation.
\newblock {\em Journal of Neuroscience}, 31(1):300--313, Jan. 2011.

\bibitem{LeadFields}
G.~Nolte and G.~Dassios.
\newblock Analytic expansion of the {EEG} lead field for realistic volume
  conductors.
\newblock {\em Phys. Med. Biol.}, 50:3807--3823, 2005.

\bibitem{LinearEEG}
L.C. Parra, C.D. Spence, A.D. Gerson, and P.~Sajdac.
\newblock Recipes for the linear analysis of {EEG}.
\newblock {\em NeuroImage}, 28:326--341, March 2005.

\bibitem{timeRelatedAudio}
T.~Paus, R.J. Zatorre, N.~Hofle, Z.~Caramanos, J.~Gotman, M.~Petrides, and A.C.
  Evans.
\newblock Time-related changes in neural systems underlying attention and
  arousal during the performance of an auditory vigilance task.
\newblock {\em Journal of Cognitive Neuroscience}, 9:392--408, 2006.

\bibitem{Econ2}
K.~Pawelzik, J.~Kohlmorgen, and K.-R. M{\"u}ller.
\newblock Annealed competition of experts for a segmentation and classification
  of switching dynamics.
\newblock {\em Neural Computation}, 8(2):340--356, 1996.

\bibitem{originalPCA}
K.~Pearson.
\newblock On lines and planes of closest fit to systems of points in space.
\newblock {\em Philosophical Magazine}, 2:559,572, 1901.

\bibitem{Pearson1901Lines}
K.~Pearson.
\newblock On lines and planes of closest fit to systems of points in space.
\newblock {\em The London, Edinburgh and Dublin Philosophical Magazine and
  Journal of Science}, 2:559--572, 1901.

\bibitem{EEG_chaos}
J.P. Pijn, J.V. Neerven, A.~Noest, and F.H. Lopes~da Silva.
\newblock Chaos or noise in {EEG} signals; dependence on state and brain site.
\newblock {\em Electroencephalography and Clinical Neurophysiology}, 79(5):371
  -- 381, 1991.

\bibitem{GroupWiseSSA}
W.~Samek, M.~Kawanabe, and C.~Vidaurre.
\newblock Group-wise stationary subspace analysis - a novel method for studying
  non-stationarities.
\newblock In G.~R. M{\"u}ller-Putz, R.~Scherer, M.~Billinger, A.~Kreilinger,
  V.~Kaiser, and C.~Neuper, editors, {\em Proceedings of the 5th International
  BCI Conference - Graz}, pages 16--20. Verlag der Technischen Universit{\"a}t
  Graz, 2011.

\bibitem{sCSP}
W.~Samek, C.~Vidaurre, K.-R. M\"uller, and M.~Kawanabe.
\newblock Stationary common spatial patterns for brain computer interfacing.
\newblock {\em Journal of Neural Engineering}, 9:026013, 2012.

\bibitem{BCIadapt}
P.~Shenoy, M.~Krauledat, B.~Blankertz, R.P.N Rao, and K.-R. M\"uller.
\newblock Towards adaptive classification for {BCI}.
\newblock {\em Journal of Neural Engineering}, 3:13--23, 2006.

\bibitem{Sugiyama:2007:CSA:1314498.1390324}
M.~Sugiyama, M.~Krauledat, and K.R. M\"{u}ller.
\newblock Covariate shift adaptation by importance weighted cross validation.
\newblock {\em J. Mach. Learn. Res.}, 8:985--1005, December 2007.

\bibitem{oai:eprints.pascal-network.org:5102}
C.~Vidaurre, M.~Kawanabe, P.~von B\"unau, B.~Blankertz, and K.-R. M\"uller.
\newblock Towards an unsupervised adaptation of {LDA} for brain-computer
  interfaces.
\newblock {\em IEEE Transactions on Biomedical Engineering}, 58:587--597, 2011.

\bibitem{carmenNeuralComp}
C.~Vidaurre, C.~Sannelli, K.-R. M\"uller, and B.~Blankertz.
\newblock Machine-learning-based coadaptive calibration for brain-computer
  interfaces.
\newblock {\em Neural Computation}, 23:791--816, 2011.

\bibitem{PRL:SSA:2009:DISCRIM}
P.~von B\"unau, F.C. Meinecke, F.J. Kir\'aly, and K.-R. M\"uller.
\newblock Finding stationary subspaces in multivariate time series.
\newblock {\em Phys. Rev. Lett.}, 103(21):214101, Nov 2009.

\bibitem{PRL:SSA:2009}
Paul von B\"unau, Frank~C. Meinecke, Franz~J. Kir\'aly, and Klaus-Robert
  M\"uller.
\newblock Finding stationary subspaces in multivariate time series.
\newblock {\em Phys. Rev. Lett.}, 103(21):214101, Nov 2009.

\bibitem{conf/icassp/WojcikiewiczVK11}
W.~Wojcikiewicz, C.~Vidaurre, and M.~Kawanabe.
\newblock Stationary common spatial patterns: Towards robust classification of
  non-stationary {EEG} signals.
\newblock In {\em ICASSP}, pages 577--580. IEEE, 2011.

\end{thebibliography}

\end{document}